\documentclass[11pt,letterpaper]{article}

\usepackage{preamble}
\usepackage{authblk}
\usepackage{orcidlink}

\title{EF(X) Orientations: A Parameterized Complexity Perspective}

\author[1,2]{Sotiris Kanellopoulos \orcidlink{0009-0006-2999-0580} }

\author[3]{Edouard Nemery \orcidlink{0009-0007-6977-9330} }

\author[1,2]{Christos Pergaminelis \orcidlink{0009-0009-8981-3676} }

\author[2,4]{\\ Minas Marios Sotiriou \orcidlink{0009-0007-3016-980X} }

\author[3]{Manolis Vasilakis \orcidlink{0000-0001-6505-2977} }

\affil[1]{National Technical University of Athens, Greece}
\affil[2]{Archimedes, Athena Research Center, Greece}
\affil[3]{Universit\'{e} Paris-Dauphine, PSL University, CNRS UMR7243, LAMSADE, Paris, France}
\affil[4]{Athens University of Economics and Business}

\affil[ ]{\{s.kanellopoulos, chr.pergaminelis\}@athenarc.gr, edouard.nemery@dauphine.psl.eu, minas\_marios@outlook.com, emmanouil.vasilakis@dauphine.eu}

\date{}

\begin{document}

\maketitle

\begin{abstract}
    The concept of fair orientations in graphs was introduced by Christodoulou, Fiat, Koutsoupias and Sgouritsa~\cite{EFXgraphs}, naturally modeling fair division scenarios in which resources are only contested by neighbors. In this model, vertices represent agents and undirected edges represent goods; edges have to be oriented towards one of their endpoints, i.e., allocated to one of their adjacent agents. Although EFX orientations (\emph{envy-free up to any good}) have been extensively studied in this setting, EF orientations (\emph{envy-free}) remain unexplored. In this work, we initiate their study primarily under the lens of parameterized complexity, presenting various tractable cases, hardness results, and parameterizations. Our results concern both simple graphs and multigraphs. Interestingly, many of our results transfer to EFX orientations, thus complementing and improving upon previous work; notably, we settle an open question regarding the structural parameterized complexity of the latter problem on graphs of polynomially-bounded valuations (Deligkas, Eiben, Goldsmith and Korchemna~\cite{ijcai/DeligkasEGK25}).
We also show that EF orientations are tractable in cases in which EFX orientations are not, particularly for binary valuations. We achieve this through a complete characterization for EF orientations with binary valuations, roughly stating that an EF orientation exists if and only if every connected component of the graph satisfies one among four graph properties. Lastly, we consider \emph{charity} in the orientation setting, establishing algorithms for finding the minimum amount of edges that have to be removed from a graph in order for EF(X) orientations to exist.

\end{abstract}

\clearpage

\section{Introduction}\label{sec:introduction}




    The fair division of indivisible goods is a fundamental problem in computer science, with numerous applications, such as distribution of inheritance, scheduling jobs, and allocating computational resources to machine learning models. The concept of fair division of \emph{divisible} resources was introduced by Steihaus~\cite{h__steihaus_1948}; in recent years, interest has shifted towards \emph{indivisible} goods, due both to the aforementioned applications and to several intriguing open questions surrounding the study of \emph{envy}. Since envy-free (EF) allocations are not guaranteed to exist, research has focused on \emph{envy minimization} (e.g.,~\cite{ISAAC_envy_ratio,LiptonEtAl,Nguyen_Rothe}),
    and \emph{relaxations} of EF, such as EF1 \cite{budish} and EFX \cite{EFXCara,Laurent_EFX}. The existence of EFX allocations remains a major open question in fair division (cf.~\cite{FD_survey}),
    with some results for two~\cite{PlautRough} and three agents~\cite{akrami2022efxallocationssimplificationsimprovements,CGM24}, and further work limiting agents' valuations~\cite{TwoValuedInstanses,babaioff2020fairtruthfulmechanismsdichotomous,DBLP:conf/aaai/HosseiniSVX21,efxexistsfor3typesofadditiveagents}.

    A restricted version of the problem was introduced by Christodoulou, Fiat, Koutsoupias and Sgouritsa~\cite{EFXgraphs}, with a graph model where the vertices represent agents and the edges represent goods; it was proven in the same paper that EFX allocations always exist in this restricted version for simple graphs. This model has been extensively studied since, with the existence of EFX allocations also being studied in multigraphs \cite{AfshinmehretTriangleFree,AfshinmehrDKMR25,bhaskar_et_al:LIPIcs.FSTTCS.2025.15_efxallocationsmultigraphclasses,kaviani2024envyfreeallocationindivisiblegoods,OntheExistenceofEFXAllocationsinMultigraphs}. An interesting concept in this model is that of \emph{EFX orientations}, i.e., EFX allocations where edges have to be allocated to one of their endpoints, naturally modeling practical applications where resources are contested only by neighbors. EFX orientations have also received significant attention in recent years, with their parameterized complexity in particular being studied in \cite{AfshinmehrDKMR25,sagt/BlazejGRS25,ijcai/DeligkasEGK25}.

    In spite of this, EF orientations remain largely unexplored.
    In this work we study them primarily under the lens of parameterized complexity, and present a plethora of tractable special cases, parameterized algorithms, and hardness results, including a complete characterization for the case of binary valuations (\cref{theorem:connected_component_props}). Many of our results transfer to EFX orientations, thus complementing and improving upon results from previous work on EFX orientations~\cite{AfshinmehrDKMR25,sagt/BlazejGRS25,ijcai/DeligkasEGK25}; see \cref{sec:contributions} for details.

    Another interesting relaxation of EF(X) that we consider in this work concerns allocations in which not every item is allocated to an agent, i.e., EF(X) with \emph{charity} \cite{CGH19}.
    To the best of our knowledge, we are the first to consider charity in the graph setting of \cite{EFXgraphs}, establishing algorithms for finding the minimum number of edges that have to remain unoriented (i.e., donated to charity) in order for EF(X) orientations to exist.

\subsection{Related Work}

\paragraph{Allocations.} Regarding EF allocations, NP-completeness and inapproximability results are shown by Lipton, Markakis, Mossel and Saberi~\cite{LiptonEtAl}.
Recent results for the existence of EFX allocations include the following: Plaut and Roughgarden~\cite{PlautRough} show existence for two agents with general monotone valuations; Chaudhury, Garg and Mehlhorn~\cite{CGM24} for three agents with additive valuations; Akrami et al.~\cite{akrami2022efxallocationssimplificationsimprovements} improve this to a slightly more general class; Prakash, Ghosal, Nimbhorkar and Varma~\cite{efxexistsfor3typesofadditiveagents} show existence when agents have one of three additive valuations.

\paragraph{Orientations.}
The existence of EFX orientations is known to be NP-complete even for symmetric valuations by Christodoulou et al.~\cite{EFXgraphs}.
Deligkas, Eiben, Goldsmith and Korchemna~\cite{ijcai/DeligkasEGK25} prove that the problem remains NP-complete even for graphs with vertex cover size $8$ or multigraphs with $10$ vertices, while also identifying one of the few known fixed-parameter tractable cases through the parameterization by \emph{slim tree-cut width} (defined in \cite{algorithmica/GanianK24}).
Blazej, Gupta, Ramanujan and Strulo~\cite{sagt/BlazejGRS25} study the complexity of the problem in relation to how close the graph is to being bipartite,
proving that it remains NP-complete for (a subclass of) binary instances,
while also showing a linear-time algorithm for symmetric binary instances of constant treewidth.
Afshinmehr, Danaei, Kazemi, Mehlhorn and Rathi~\cite{AfshinmehrDKMR25} and Hsu~\cite{ecai/Hsu25} further study EFX orientations in (bipartite) multigraphs. 
Zeng and Mehta~\cite{ifaamas/ZengM25} relate the existence of EFX orientations to the graph's chromatic number.
EFX orientations have also been studied in the setting of chores, i.e., negatively valued items, and mixed manna, i.e., goods and chores \cite{hsu2025polynomialtimealgorithmefxorientationsChores,ijcai/ZWL024}.
Despite all the aforementioned work for EFX orientations, to the best of our knowledge,
the only work relevant to EF orientations appears in \cite{aldt/MisraS24},
where the authors study EF \emph{allocations} in simple graphs,
with a polynomial-time algorithm for binary valuations also transferring to \emph{orientations}.

\paragraph{Charity.}
The concept of charity was first introduced by Caragiannis, Gravin, and Huang~\cite{CGH19}, in the setting of EFX allocations.
Chaudhury, Kavitha, Mehlhorn and Sgouritsa~\cite{CKMS21} proved that the number of unallocated items that guarantees existence of EFX for general monotone valuations is upper bounded by $n-1$, where $n$ is the number of agents. This bound was later reduced to $n-2$ for additive valuations \cite{AlmostFullEFXforFourAgents}, and extended to general monotone valuations~\cite{mahara2021extensionadditivevaluationsgeneral}. For 4 agents, this number is 1 \cite{AlmostFullEFXforFourAgents}.
To the best of our knowledge, charity has never been studied in the context of fair orientations prior to this work, although one could obtain a simple charity bound of $|E|/2$ for EFX orientations in simple graphs by removing edges to obtain a bipartite graph~\cite{erd/Erdos65},
which is known to always admit an EFX orientation \cite{ifaamas/ZengM25}.


\subsection{Our Contributions}\label{sec:contributions}

We start by considering in \cref{subsec:EFo:linear} the {\EFo} problem%
\footnote{For formal definitions of all the problems in question we refer to \cref{sec:preliminaries}.} in multigraphs with binary valuation functions.
We provide a complete characterization of EF orientations in this case (\cref{theorem:connected_component_props}), roughly stating that an EF orientation exists if and only if every connected component of the graph satisfies at least one among four graph properties. Our characterization implies a linear-time algorithm for {\EFo} in the case of binary valuations; in fact, the algorithm can be extended to the more general {\EFmc} problem (i.e., finding the minimum amount of edges that have to be removed for an EF orientation to exist).
This result contrasts the NP-hardness of {\EFXo} in the analogous setting,
which persists even for symmetric instances on simple graphs that are two edge-deletions away from being bipartite~\cite{sagt/BlazejGRS25}. It also generalizes to multigraphs the algorithm of Misra and Sethia~\cite{aldt/MisraS24} for \EFo\ in simple graphs.

In \cref{subsec:EF:hardness} we present various hardness results for {\EFo}.
First we show that the problem remains NP-hard even on $3$-regular simple graphs with weights in $\{1,2\}$ and symmetric valuations, through a reduction from an appropriate variant of SAT.
Given this hardness for very restrictive valuations,
we next ask whether restricting the \emph{structure} of the input graph allows for more efficient algorithms.
Via a simple reduction from \textsc{Partition} we show that, for arbitrary weights, 
{\EFo} remains NP-hard even for simple graphs of vertex cover number~$2$ and multigraphs with~$2$ vertices.
Therefore, we move on to consider instances with polynomially-bounded weights, in which case
we prove W[1]-, XNLP-, and XALP-hardness for the parameterization by vertex cover, pathwidth, and treewidth respectively (\cref{thm:EF:w1h}) through a reduction from the {\TOO} problem~\cite{wg/BodlaenderCW22,dam/BodlaenderS26}. We thus argue that the problem remains highly intractable even in structured inputs.
As a byproduct of our reduction in \cref{thm:EF:w1h} and in conjunction with a recent result on the fine-grained complexity of \UBP~\cite{bringmann2026tightsethbasedlowerbounds}, we obtain a $n^{o(\vc)}$ ETH-based lower bound,
where $\vc$ denotes the vertex cover number of the input graph, that holds even for simple graphs.

Although the hardness results of \cref{subsec:EF:hardness} seem very discouraging,
we next identify restricted cases on which we are able to obtain fast algorithms.
First, we consider instances on simple graphs with at most $k$ \emph{heavy} edges, that is, edges that have valuation at least $2$ for an endpoint;
this naturally generalizes the binary-valuation regime considered in \cref{subsec:EFo:linear}, albeit restricted to simple graphs.
Here we develop a $2^k n^{\bO(1)}$ algorithm, which we complement with a tight SETH-based lower bound (\cref{thm:seth-ef,thm:EFo:Heavy}).
We view this as a first step towards generalizing the algorithm of \cref{thm:EFo:linear_algo} to an FPT algorithm parameterized by the
number of heavy edges. 
Second, using dynamic programming (DP), we show that {\EFo} can be solved in time $(n+m)^{\bO(\tw)}$ for polynomially-bounded weights,
where $\tw$ denotes the treewidth of the graph and $m$ the number of edges.
Our algorithm works for the more general {\EFmc} problem in multigraphs with non-symmetric valuations,
and matches the lower bound of \cref{thm:EF:w1h}.

{\EFo} is arguably a very natural, understudied problem and our results, apart from being interesting in their own right,
additionally pave the way to advance the state-of-the-art for the closely related, and better understood, {\EFXo} problem.
In particular, in \cref{subsec:EF->EFX} we present a reduction from {\EFo} to {\EFXo} that allows us to transfer most hardness results of \cref{subsec:EF:hardness}
to the latter. Consequently, we obtain NP-hardness for simple graphs
of vertex cover number~$4$ and multigraphs of~$4$ vertices, improving upon the hardness results of Deligkas et al.~\cite{ijcai/DeligkasEGK25} for constants~$8$ and~$10$ respectively, as well as a hardness result by Afshinmehr et al.~\cite{AfshinmehrDKMR25} for multigraphs with~$8$ vertices.
More importantly, we obtain a $n^{o(\vc)}$ ETH-based lower bound for simple symmetric instances with polynomially-bounded weights.
As a matter of fact, in \cref{subsec:EFXo:tw} we yet again employ DP and obtain an $(n+m)^{\bO(\tw)}$ algorithm,
this time for the {\EFXmc} problem; once again, our algorithm works for multigraphs and non-symmetric valuations.
Since this matches the aforementioned lower bound, we settle the open question of Deligkas et al.~\cite{ijcai/DeligkasEGK25} on the complexity of {\EFXo} for instances with
polynomially-bounded weights and small vertex cover, and paint a complete picture
on the parameterized complexity of {\EFXo} parameterized by treewidth, arguably the most well-studied structural parameter.

We refer to \cref{table:algo,table:hardness} for a summary of our algorithms and hardness results respectively.

\newcommand{\lightrule}{\arrayrulecolor{black!30}\hline\arrayrulecolor{black}}
{
\begin{table}[ht]
\centering
\begin{tabular}{|c|c!{\color{black!30}\vrule}c|}
\hline
& \textbf{Running time} & \textbf{Conditions} \\
\hline
 EF & $\bO(n+m)$ & Binary valuations \\
\lightrule
\multirow{2}{*}{EF} & \multirow{2}{*}{$2^k n^{\bO(1)}$} & At most $k$ heavy edges \\  &  & Simple graph \\
\lightrule

\multirow{2}{*}{EF(X)} & \multirow{2}{*}{$W^{\bO(\tw)}(n+m)^{\bO(1)}$}
 & $W$: \emph{max. shared weight}\\ & &  $\tw$: treewidth\\
\lightrule

\multirow{2}{*}{EF(X)} & \multirow{2}{*}{$(n+m)^{\bO(\tw)}$}
 & Poly weights \\ &  & $\tw$: treewidth \\

\hline
\end{tabular}
\caption{Our algorithms for EF and EFX \textsc{Orientation} in (multi)graphs with $n$ vertices and $m$ edges (\cref{subsec:EFo:linear,subsec:EFo_heavy,subsec:EFo_shared_weight,subsec:EFXo:tw}). All algorithms except the one in the second row also extend to the respective minimum-charity problem.}
\label{table:algo}
\end{table}
}

\begin{table}[ht]
\centering
\begin{tabular}{|c|c!{\color{black!30}\vrule}c|}
\hline
& \textbf{Hardness/} & \multirow{2}{*}{\textbf{Conditions}} \\
& \textbf{Lower Bound} & \\
\hline
 \multirow{2}{*}{EF} & no $(2-\varepsilon)^{k}n^{\bO(1)}$ time algorithm & At most $k$ heavy edges \\ & (under SETH) & Polynomial weights \\
\lightrule
\multirow{2}{*}{EF} & \multirow{2}{*}{NP-complete} & 3-regular graphs \\ & & Weights in \{1,2\}    \\
\lightrule
\multirow{2}{*}{EF} & NP-complete & Vertex cover 2 \textbf{or}\\
& (weakly) & Multigraph with 2 vertices    \\
\lightrule
\multirow{2}{*}{EFX} & NP-complete & Vertex cover 4 \textbf{or}\\
& (weakly) & Multigraph with 4 vertices    \\
\lightrule
\multirow{4}{*}{EF(X)} & W[1]-hard & Parameter: vertex cover number \\
\arrayrulecolor{black!30}\cline{2-3}\arrayrulecolor{black}
 & XNLP-hard & Parameter: pathwidth   \\
\arrayrulecolor{black!30}\cline{2-3}\arrayrulecolor{black}
 & XALP-hard & Parameter: treewidth \\ \arrayrulecolor{black!30}\cline{2-3}\arrayrulecolor{black} & & \textbf{All three:} Polynomial weights   \\
\lightrule
\multirow{2}{*}{EF(X)} & no $n^{o(\vc)}$ time algorithm & $\vc$: vertex cover number \\ & (under ETH) & Polynomial weights \\
\hline
\end{tabular}
\caption{Our hardness results for EF and EFX \textsc{Orientation} (\cref{subsec:EF:hardness,subsec:EF->EFX} respectively). All of these results hold even for \textbf{simple symmetric} instances, except the ``multigraph'' results in the third/fourth rows, which are restricted just to symmetric instances.}
\label{table:hardness}
\end{table}

\section{Preliminaries}\label{sec:preliminaries}

Throughout the paper we use standard graph notation \cite{books/Diestel17}
and we assume familiarity with the basic notions of parameterized complexity (cf. \cref{parameterized_complexity}). 
For a graph $G$ and $S \subseteq V(G)$ a subset of its vertices,
$G[S]$ denotes the subgraph induced by $S$ while $G - S$ denotes $G[V \setminus S]$.
For $x, y \in \mathbb{Z}$, let $[x, y] = \setdef{z \in \mathbb{Z}}{x \leq z \leq y}$ while $[x] = [1,x]$.
Proofs of statements marked with $(\appsymb)$ are deferred to the appendix.

Graphs in the context of this paper are \emph{multigraphs}, i.e., parallel edges between two vertices may exist, unless we explicitly state that a graph is \emph{simple}. All graphs considered are undirected and without loops.
We denote the set of edges between vertices $u,v$ of a multigraph as $E_{uv}$, and the set of edges incident to $u$ as $E_u$.

\subsection{Model and Problem Definitions}

We are given a set of agents $V=\{1,\ldots,n\}$ and a set of items $E$ with $|E|=m$. Each agent $i \in [n]$ has a valuation function $v_i \colon 2^E \to \mathbb{N}$. An \emph{allocation} is a partition of $E$ into bundles $X_1,\ldots,X_n$, where agent $i$ receives the bundle $X_i$. A \emph{partial allocation} is an allocation of some subset of $E$. 

In the context of this work, we only consider \emph{additive} valuation functions, as is standard for most of the fair orientation literature:
each $v_i$ is defined over $E$, and can be extended to $2^E$ with $v_i(X):=\sum_{e\in X}v_i(e)$, $\forall X\subseteq E$. As is standard, we assume $v_i(\varnothing)=0$ for all $i\in V$.

\begin{definition}[Envy]\label{def:envy}
    Given an allocation $X_1,\ldots,X_n$, we say that agent $i$ \emph{envies} $j$ if
    $v_i(X_i) < v_i(X_j)$. We say that $i$ \emph{strongly envies} $j$ if $v_i(X_i) < v_i(X_j \setminus e)$ for some $e \in X_j$.
\end{definition}

\begin{definition}[Fair allocations]\label{def:fair_alloc}
    We say that an allocation is \emph{envy-free (EF)} if no agent envies another agent. We say that an allocation is \emph{envy-free up to any good (EFX)} if no agent strongly envies another agent. 
\end{definition}

In this paper, we consider the model introduced by \cite{EFXgraphs}: we are given sets of agents $V$ and items $E$ that form a multigraph $G=(V,E)$, where for all edges $e\in E_{ij}$ $(i,j \in V)$ it holds that $v_u(e)=0$, $\forall u\in V\setminus\{i,j\}$.\footnote{For simplicity, we use the notation $v_u(e)$ instead of $v_u(\{e\})$ to denote the valuation of a single edge $e$ for vertex $u$.} Throughout the paper we may call the agents \textit{vertices} and the goods \textit{edges}. We denote $n=|V|$ and $m=|E|$.

An \emph{orientation} is an allocation $X_1,\ldots,X_n$, where for all $e\in E_{ij}$ $(i,j \in V)$ it holds that $e\in X_i$ or $e\in X_j$. A \emph{partial orientation} is an orientation of some subset of~$E$. Given an orientation, we say that $e\in E_{ij}$ $(i,j \in V)$ is oriented towards~$i$ if $e\in X_i$, i.e., if $i$ receives~$e$. An EF(X) orientation is an EF(X) allocation that is an orientation.

We now define the problems which we study in this work.

\defproblem{\EFo}
{A multigraph $G=(V,E)$ and values $v_i(e), v_j(e)$ for each $e\in E_{ij}$ $(i,j \in V)$.}
{Determine whether $G$ admits an EF orientation.}\label{def:EFo}

\defproblem{\EFmc}
{A multigraph $G=(V,E)$ and values $v_i(e), v_j(e)$ for each $e\in E_{ij}$ $(i,j \in V)$.}
{Find the smallest integer $k\geq 0$ s.t.~$G$ admits a partial EF orientation of size $|E|-k$.} \label{def:mc_EFo}

Intuitively, the \EFmc\ problem asks for the minimum amount of edges that must be removed from a graph for an EF orientation to exist. Note that any algorithm for \EFmc\ also answers \EFo, hence, throughout the paper we may omit an \EFo\ algorithm and directly state the (more general) \EFmc\ algorithm.

The \EFXo\ and \EFXmc\ problems are defined similarly.

We call an instance of any of these problems
\begin{itemize}
    \item \emph{simple}, if for all $i,j \in V$: $|E_{ij}|\leq 1$,
    \item \emph{symmetric}, if for all $i,j \in V$ and all $e\in E_{ij}$: $v_i(e)=v_j(e)$. In that case, it suffices to store one valuation function $v \colon E \to \mathbb{N}$, where for all $i,j \in V$ and all $e\in E_{ij}$: $v(e)=v_i(e)=v_j(e)$.
\end{itemize}

Throughout the paper, we may refer to the valuation functions in the context of graphs as \emph{weights}: each edge is represented by two weights in the general case (one for each of its endpoints), however one weight per edge suffices if the instance is symmetric.

\begin{toappendix}
    \section{Parameterized Complexity}\label{parameterized_complexity}

In this section we briefly review some basic notions of parameterized complexity
which we use throughout the paper;
for a more thorough exposition we refer to standard textbooks, e.g.,~\cite{books/CyganFKLMPPS15,books/DowneyF13}.

An instance of a \emph{parameterized problem} is a pair
$(x,k)$, where $x$ is the main input (with size $n := |x|$) and
$k \in \mathbb{N}$ is the parameter.
We say that a parameterized problem is \emph{fixed-parameter tractable} (FPT)
and belongs to the class FPT if it admits an algorithm of running time
$f(k) \cdot n^{\bO(1)}$ for some computable function $f$.

Let $\Pi_1$ and $\Pi_2$ be two parameterized problems.
A \emph{parameterized reduction} (also known as an \emph{FPT reduction}) from $\Pi_1$ to $\Pi_2$ is an algorithm that, given an instance $(x,k)$ of $\Pi_1$,
produces an instance $(x',k')$ of $\Pi_2$ such that:
\begin{itemize}
    \item $(x,k)$ is a yes-instance of $\Pi_1$ if and only if $(x',k')$ is a yes-instance of $\Pi_2$,
    \item $k' \leq g(k)$ for some computable function $g$,
    \item the running time of the algorithm is $f(k) \cdot {|x|}^{\bO(1)}$ for some computable function $f$.
\end{itemize}
Notice that if $\Pi_2$ is fixed-parameter tractable (FPT), then so is $\Pi_1$.

We will use parameterized reductions to show hardness results for parameterized problems.
In particular, we define the class W[1] as the set of parameterized problems that are
FPT-reducible to \textsc{$k$-Clique} parameterized by $k$.%
\footnote{The \textsc{$k$-Clique} problem asks,
given a graph $G$ and an integer $k$, whether $G$ contains a clique of size at least $k$.}
It holds that FPT $\subseteq$ W[1], and it is widely believed that FPT $\neq$ W[1],
thus showing that a problem is W[1]-hard is strong evidence against fixed-parameter tractability~\cite{books/CyganFKLMPPS15,books/DowneyF13}.

Another complexity class considered in this paper is XNLP.
This class was recently introduced by~\cite{iandc/BodlaenderGNS24},
and it consists of parameterized problems decidable by a nondeterministic algorithm that,
on inputs of size $n$ with parameter $k$, uses time $f(k) \cdot n^{\bO(1)}$ and space
$f(k) \cdot \log n$, for some computable function $f$.
It holds that W[1] $\subseteq$ XNLP,
and in fact XNLP-hardness implies W[$t$]-hardness for all $t \in \mathbb{N}$.
In order to show XNLP-hardness results, we need to introduce the notion of
\emph{parameterized logspace reductions}, which are simply
parameterized reductions that use $\bO(g(k) + \log |x|)$ space,
where $g$ is a computable function.
In that case, if $\Pi_1$ is XNLP-hard and there is a parameterized logspace reduction from $\Pi_1$ to $\Pi_2$,
then $\Pi_2$ is also XNLP-hard~\cite{iandc/BodlaenderGNS24}.
Another class we use is XALP, which can be seen as the ``tree variant'' of XNLP,
where we additionally allow access to an auxiliary stack (with only top element lookup allowed)~\cite{iwpec/BodlaenderGJPP22}.
This class contains XNLP, and XALP-hardness can be shown via parameterized logspace reductions as well.
XNLP and XALP have been shown to be the ``natural home'' of many problems parameterized by pathwidth and treewidth respectively,
as many problems tend to be complete for each class respectively.

The \emph{treewidth} of a graph is a measure of how close the graph is to being a tree, and is arguably the most well-studied structural parameter.
A \emph{tree-decomposition} $\mathcal{T}$ of a graph $G$ is a pair $(T,\setdef{B_t}{t \in V(T)})$ where
$T$ is a tree and each node $t$ of $T$ is assigned a \emph{bag} $B_t \subseteq V(G)$,
such that: (i) every vertex $v \in V(G)$ is contained in some bag $B_t$; (ii) for every edge $\{u,v\} \in E(G)$ there exists $t \in V(T)$ with
$\{u,v\} \subseteq B_t$; and
(iii) for every $v \in V(G)$, the set $\setdef{t \in V(T)}{v \in B_t}$ induces a connected subtree of $T$.
The \emph{width} is $\max_{t \in V(T)} |B_t|-1$;
the \emph{treewidth} $\tw(G)$ is the minimum width over all tree-decompositions of $G$.
A tree-decomposition is \emph{nice} if each node of $T$ has one of the following types:
(i) \emph{Leaf} nodes $t$ where $B_t = \varnothing$;
(ii) \emph{Introduce} nodes $t$ with a unique child $t'$ such that $B_t = B_{t'} \cup \{v\}$ for some $v \in V(G)$;
(iii) \emph{Forget} nodes $t$ with a unique child $t'$ such that $B_{t'} = B_t \cup \{v\}$ for some $v \in V(G)$;
(iv) \emph{Join} nodes $t$ with exactly two children $t_1,t_2$ such that $B_t = B_{t_1} = B_{t_2}$.
It is known that given any tree
decomposition one can in linear time construct a nice tree decomposition of the same width.
Computing a tree-decomposition of minimum width is in FPT when parameterized by the treewidth~\cite{siamcomp/Bodlaender96,books/Kloks94,stoc/KorhonenL23},
and even more efficient algorithms exist for obtaining near-optimal tree-decompositions~\cite{focs/Korhonen21}.

We conclude this section by recalling two well-known hypotheses from the area of
fine-grained complexity~\cite{jcss/ImpagliazzoP01,jcss/ImpagliazzoPZ01}.
The \emph{Exponential Time Hypothesis} (ETH) posits that there exists a
constant $c>0$ such that no algorithm solves 3-SAT on $n$ variables and $m$
clauses in time $2^{c (n+m)} \cdot (n + m)^{\bO(1)}$.
The \emph{Strong Exponential Time Hypothesis} (SETH) states that for every
$\varepsilon>0$ there exists a constant $k \ge 3$ (depending only on $\varepsilon$)
such that no algorithm solves $k$-SAT on
$n$ variables in time $(2 - \varepsilon)^{n} \cdot n^{\bO(1)}$.
Both hypotheses constitute standard tools to obtain oftentimes tight (conditional) lower bounds.
With respect to parameterized complexity, they have been used to pinpoint the best parametric dependence possible for various parameterized algorithms~\cite{talg/LokshtanovMS18,siamcomp/LokshtanovMS18}.

\end{toappendix}

\section{EF Orientations}\label{section:EFo}

In this section we present our results concerning EF orientations. The following observation will be useful throughout the section.

\begin{observation}\label{obs:zero_value_goods}
    Edges with zero value for both endpoints can be removed from the input without affecting the existence of EF orientations.
    Similarly, edges with zero value for one endpoint and non-zero value for the other can always be oriented towards the latter without affecting the existence of an EF orientation.\footnote{Note that neither of these properties hold for EFX orientations.}
\end{observation}


\subsection{Algorithms for Binary Valuations}\label{subsec:EFo:linear}

In this subsection we restrict the \EFo\ problem to binary valuations, i.e., we assume that $v_i(e),\ v_j(e) \in \{0,1\}$ for all $e\in E_{ij}$, $i,j\in V$.
By \cref{obs:zero_value_goods}, we remove from the input all edges of weight $0$ for both endpoints and orient all edges of weight $0$ for one endpoint and weight~$1$ for the other endpoint towards the latter.

\cref{obs:happy_vertex} is the key to proving (part of) our main theorem in this subsection (\cref{theorem:connected_component_props}), which is essentially a complete characterization of EF orientations with binary valuations.

\begin{observation}\label{obs:happy_vertex}
    For an \EFo\ instance with binary valuations, define $E'\subseteq E$ to be the subset of the edges with weight $1$ for both their endpoints. Suppose we have a partial EF orientation where some vertex $i$ has received an edge $e^*\in E_{ij}$ s.t. $v_i(e^*)=1$. Then, all edges $e\in E'_{ik}$ for $k\neq j$ can be oriented as follows while preserving envy-freeness. Let $p=|E'_{ik}|$.
    \begin{itemize}
        \item If $p$ is even, orient $p/2$ edges towards $i$ and $p/2$ edges towards $k$. It is clear that neither $i$ nor~$k$ envy each other.
        \item If $p$ is odd, orient $(p-1)/2$ edges towards $i$ and $(p-1)/2 +1$ edges towards $k$. It is clear that $k$ does not envy~$i$; however, it also holds that $i$ does not envy $k$ because $i$ has received $e^*\in E_{ij}$ with $v_i(e^*)=1$ and $j\neq k$ by assumption.
    \end{itemize}
    Additionally, if $v_j(e^*)=0$, then the same observation applies even for $k = j$.
\end{observation}

\begin{theoremrep}[Characterization \appsymb]\label{theorem:connected_component_props}
    Let $\mathcal{I}=(G=(V,E),\{v_i\}_{i\in V})$ be an instance of \EFo\ with binary valuations and define $G'=(V,E')$ as the subgraph of $G$ s.t.~$E'\subseteq E$ consists only of the edges in $E$ with weight $1$ for both their endpoints. $\mathcal{I}$ has an EF orientation if and only if for every connected component $C$ of $G'$ at least one of the following properties holds.
    \begin{enumerate}
        \item $C$ contains a circuit\footnote{In graph theory, a circuit is defined as a closed trail where a sequence of vertices starts and ends at the same vertex without repeating any edges. Informally, it is a cycle that may pass through each vertex more than once, while only using each edge at most once.} consisting of at least $3$ vertices.\label{prop1}
        \item $C$ contains a vertex $i$ incident with some edge $e\in E_{ij}$ s.t. $v_i(e)=1$ and $v_j(e)=0$.\label{prop2}
        \item $C$ contains two vertices $i,j$ s.t. $|E'_{ij}|$ is even and non-zero.\label{prop3}
        \item $C$ consists of only one vertex.\label{prop4}
    \end{enumerate}
\end{theoremrep}

\begin{proof}
    We will prove the forward direction of the theorem by constructing an EF orientation given each of the four properties.
    As already explained, we orient all edges $e\in E_{ij}$ s.t. $v_i(e)=1$ and $v_j(e)=0$ towards $i$ (for all $i,j\in V$), and remove all edges with weight $0$ for both endpoints. From now on we may assume that all unoriented edges have weight $1$ for both of their endpoints (i.e., all unoriented edges are in the set $E'$), hence $G$ has an EF orientation if and only if every connected component of $G'=(V,E')$ has an EF orientation, granted the aforementioned orientation for all edges in $E \setminus E'$.
    We will prove that if any of the four properties hold for some connected component $C$ of $G'$, then $C$ has an \EFo, by repeatedly applying \cref{obs:happy_vertex}.

    \textbf{Property 1.} Let $C=(V_c,E_c)$ be a connected component of $G'$ satisfying Property~\ref{prop1}, i.e., $C$ contains a circuit $B=(V_b,E_b)$ with $|V_b|\geq 3$. By the definition of circuit, there exists a cyclic route that visits all edges in $E_b$ exactly once, starting and ending at the same vertex in $V_b$. Pick such a route and orient all edges in $E_b$ along the route; this way, every vertex in $V_b$ receives at least one edge (see \cref{fig:circuit}).
    \begin{figure}[ht]
        \centering
        \includegraphics[width=0.3\linewidth]{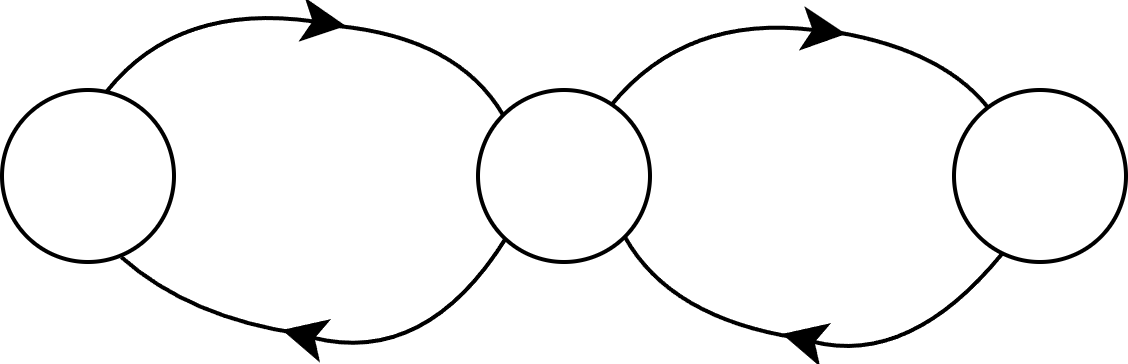}
        \caption{An example of a circuit consisting of $3$ vertices, without containing a cycle of length greater than $2$. Observe that no vertex is envious of another if the edges are oriented along a cyclic route as shown (assuming all edges have weight $1$ for both endpoints).}
        \label{fig:circuit}
    \end{figure}
    Let $i,j$ be any two vertices in $V_b$. By the above, it holds that either $i$ or $j$ has now received some edge $e \notin E_{ij}$ with weight $1$ (if both $i$ and $j$ had only received edges in $E_{ij}$, that would imply that $|V_b|=2$, contradiction). Hence, we can apply \cref{obs:happy_vertex} to orient all unoriented edges in $E'_{ij}$ while preserving envy-freeness (for all $i,j \in V_b$). Since every vertex in $V_b$ has received an edge in $E_b$, we can also apply \cref{obs:happy_vertex} to orient all edges between vertices in $V_b$ and their neighbors that are not in $V_b$, which forces the neighbors of the vertices in $V_b$ to receive at least one edge of weight $1$.
    This renders \cref{obs:happy_vertex} applicable to orient the edges between those vertices and their own neighbors, and so on. This orients all $e\in E_c$ while preserving envy-freeness.

    \textbf{Property 2.} Let $C=(V_c,E_c)$ be a connected component of $G'$ satisfying Property~\ref{prop2}, i.e., there is a vertex $i\in V_c$ that has received some edge $e\in E_{ij}$ s.t. $v_i(e)=1$ and $v_j(e)=0$, for some $j \in V$ (see \cref{fig:external_edge_on_multitree}). This renders \cref{obs:happy_vertex} applicable\footnote{Since $v_j(e)=0$, \cref{obs:happy_vertex} is applicable for all neighbors of $i$, including $j$.} to orient all edges incident with $i$, and then we can apply it iteratively to orient the edges incident with the neighbors of $i$, and so on, as we did for Property~\ref{prop1}, thus orienting all $e\in E_c$ while preserving envy-freeness.

    \begin{figure}[ht]
        \centering
        \includegraphics[width=0.5\linewidth]{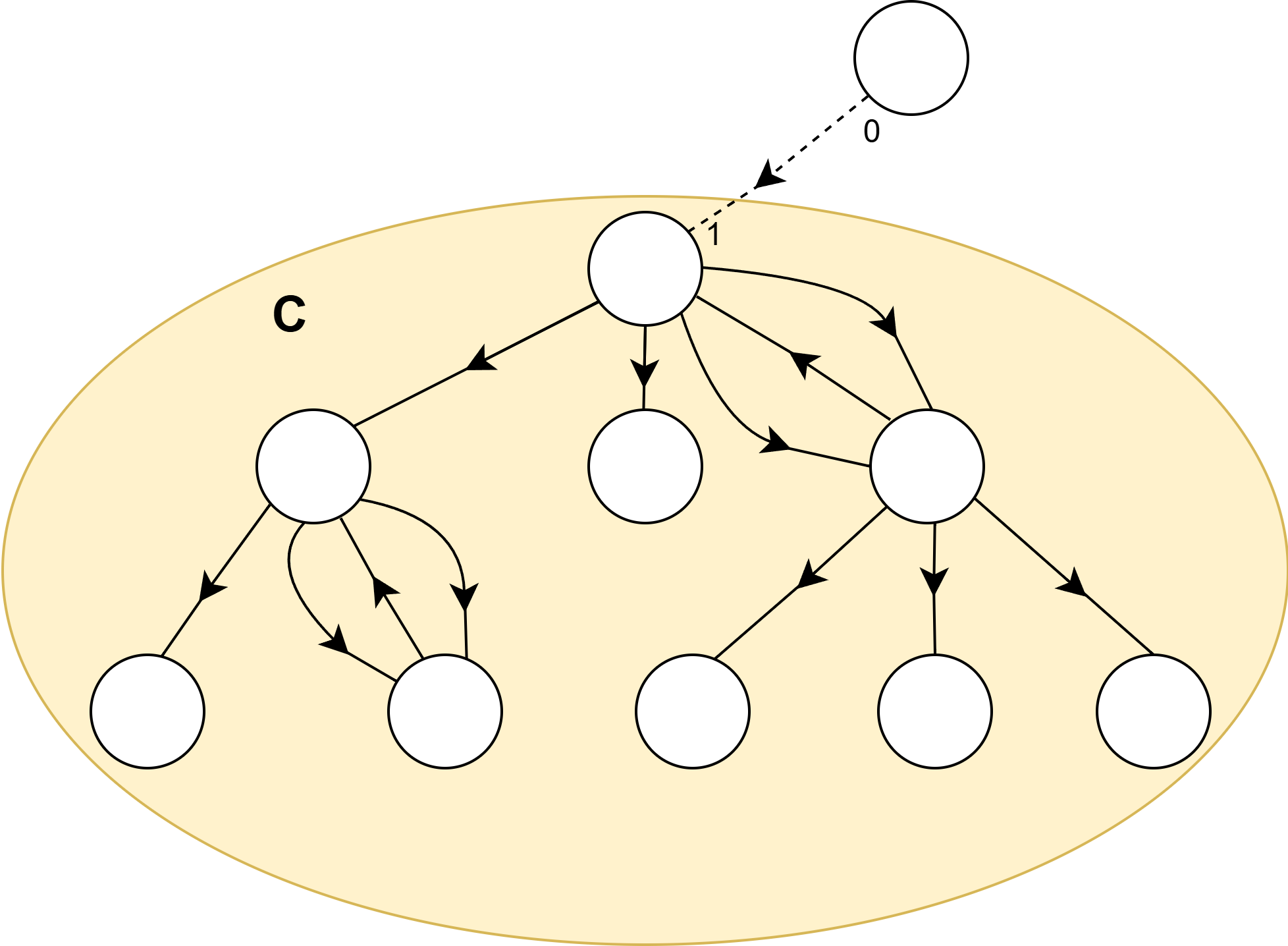}
        \caption{An example of a connected component $C$ of $G'$ that satisfies Property~\ref{prop2} of \cref{theorem:connected_component_props}, i.e., some vertex receives an edge of weight~$1$ for itself and weight $0$ for the other endpoint. This allows us to repeatedly apply \cref{obs:happy_vertex} starting from that vertex, orienting all edges of $C$ to obtain an EF orientation as demonstrated.}
        \label{fig:external_edge_on_multitree}
    \end{figure}

    \textbf{Property 3.} Let $C=(V_c,E_c)$ be a connected component of $G'$ satisfying Property~\ref{prop3}, i.e., there are vertices $i,j \in V_c$ s.t. $|E'_{ij}|>0$ is even. We orient half of these edges towards $i$ and the other half towards $j$, thus ensuring that $i$ and $j$ do not envy each other and each of them has received an edge of weight $1$. Now we can apply \cref{obs:happy_vertex} to orient the unoriented edges between $i$ and its neighbors, and then their own neighbors, and so on (same for $j$ and its neighbors), as we did for Property~\ref{prop1}. This orients all $e\in E_c$ while preserving envy-freeness.

    \textbf{Property 4.} In this case, the empty set is a trivial EF orientation for $C$.

    This concludes the forward direction of the proof. It remains to prove the converse, i.e., that if for some connected component $C=(V_c,E_c)$ of $G'$ none of the four properties hold, then $C$ has no EF orientation.

    \textbf{Converse.} Let $C=(V_c,E_c)$ be a (multi)graph that violates all four properties. Then:
    \begin{itemize}
        \item $C$ is a tree that allows an odd\footnote{Non-zero even numbers of parallel edges are not allowed, due to the negation of Property~\ref{prop3}.} number of parallel edges between neighbors (instead of just one).
        \item For every $i \in V_c$ there is at most one\footnote{If some $v$ was connected with two other vertices through a number of edges greater than one, that would induce a circuit consisting of $3$ vertices, satisfying Property~\ref{prop1}. See \cref{fig:circuit} for an example.} vertex $j\in V_c$ s.t. $|E_{ij}|>1$.
        \item For every $i \in V_c$ and every $j \in V$ (note that $j$ does not have to be in $V_c$) there is no $e\in E_{ij}$ s.t. $v_i(e)=1$ and $v_j(e)=0$. Thus, every edge in $E$ incident to some vertex in~$V_c$ has weight~$1$ for both endpoints.
        \item $|E_c| \geq 1$.
    \end{itemize}

    Fix an arbitrary root $r$ for the tree $C$ and let $L_1 =\{r\}$ be the first \emph{level} of the tree, let $L_2$ be the second level consisting of the neighbors of $r$, and so on, up to level $L_{\max}$. For $v\in V_c$ we define $s(v)$ as the number of edges connecting $v$ to the previous level of the tree. Note that, by the assumptions for $C$, each $v\in V$ (excluding the root $r$) is connected to exactly one vertex of the previous level, with an odd amount of edges. We will prove the following claim through induction on the levels $L_i$.

    \begin{claim}\label{claim:induction_tree}
        Every $v\in L_{i}$ must receive at least $(s(v)-1)/2 +1$ of the $s(v)$ edges connecting it to a vertex of level $L_{i-1}$, for all $i=\{2,\ldots,\max\}$.
    \end{claim}

    \begin{nestedproof}

    \emph{Induction basis.} Let $i\in L_{\max}$ be a leaf and $j\in L_{\max-1}$ be its (only) neighbor. If $|E_{ij}|=1$, then it is clear that the edge between $i$ and $j$ must be oriented towards $i$ in any EF orientation (otherwise, $i$ would envy $j$). Similarly, if $|E_{ij}|>1$, then $|E_{ij}|$ must be odd and at least $(|E_{ij}|-1)/2 +1$ edges must be oriented towards $i$.

    \emph{Induction hypothesis.} Suppose that every $v\in L_{k+1}$ ($k\geq 2$) must receive at least $(s(v)-1)/2 +1$ of the $s(v)$ edges connecting it to a vertex of level $L_k$.

    \emph{Inductive step.} We now prove that every $i\in L_k$ of the tree must receive at least $(s(i)-1)/2 +1$ of the $s(i)$ edges connecting it to a vertex of level $L_{k-1}$. Consider three cases:
    \begin{itemize}
        \item Case 1: $|E_{ij}|\leq 1$ for all $j\in V_c$. By induction hypothesis, $i$ receives none of the edges connecting it to level $L_{k+1}$, hence it must receive the (only) edge connecting it to level $L_{k-1}$, as it would otherwise envy its neighbors in $L_{k+1}$.
        \item Case 2: $|E_{ij}|>1$ for some $j\in L_{k+1}$. By induction hypothesis, $j$ must receive at least $(|E_{ij}|-1)/2 +1$ of the edges in $E_{ij}$, hence $i$ may receive (at most) $(|E_{ij}|-1)/2$ of these edges. By the assumptions for $C$, it holds that $|E_{iv}| \leq 1$, $\forall v\neq j,\ v\in V_c$, hence $i$ receives no other edges between itself and level $L_{k+1}$, and is connected to level $L_{k-1}$ by a single edge. This edge must be oriented towards $i$, as otherwise $i$ would envy $j$.
        \item Case 3: $|E_{ij}|>1$ for some $j\in L_{k-1}$. By the assumptions for $C$, it holds that $|E_{iv}| \leq 1$, $\forall v\neq j,\ v\in V_c$. Thus, by induction hypothesis, $i$ receives none of the edges between itself and level $L_{k+1}$. In order for $i$ to not envy $j$, at least $(|E_{ij}|-1)/2 +1$ of the edges in $E_{ij}$ have to be oriented towards $i$.
    \end{itemize}
    This concludes the proof of the claim.
    \end{nestedproof}
    We are now ready to prove the converse direction of the theorem. By the assumptions for~$C$, it holds that the root $r$ has at most one neighbor $j\in L_{2}$ s.t. $|E_{rj}|>1$. First, suppose such a neighbor $j$ exists. By \cref{claim:induction_tree}, at least $(|E_{rj}|-1)/2 +1$ of the edges in $E_{rj}$ must be oriented towards $j$ and all other edges incident to $r$ must be oriented towards the respective vertices in $L_{2}$. This causes $r$ to envy $j$. Now, suppose $j$ did not exist, i.e., $|E_{ri}|\leq 1,\ \forall i\in V_c$. For this case, \cref{claim:induction_tree} implies that $r$ cannot receive any edges at all, causing it to envy all of its neighbors in $L_2$. In both cases, $r$ is envious of some other vertex, hence there is no EF orientation for $C$ (and, thus, neither for~$G$).
\end{proof}

\begin{figure}[ht]
        \centering
        \includegraphics[width=0.5\linewidth]{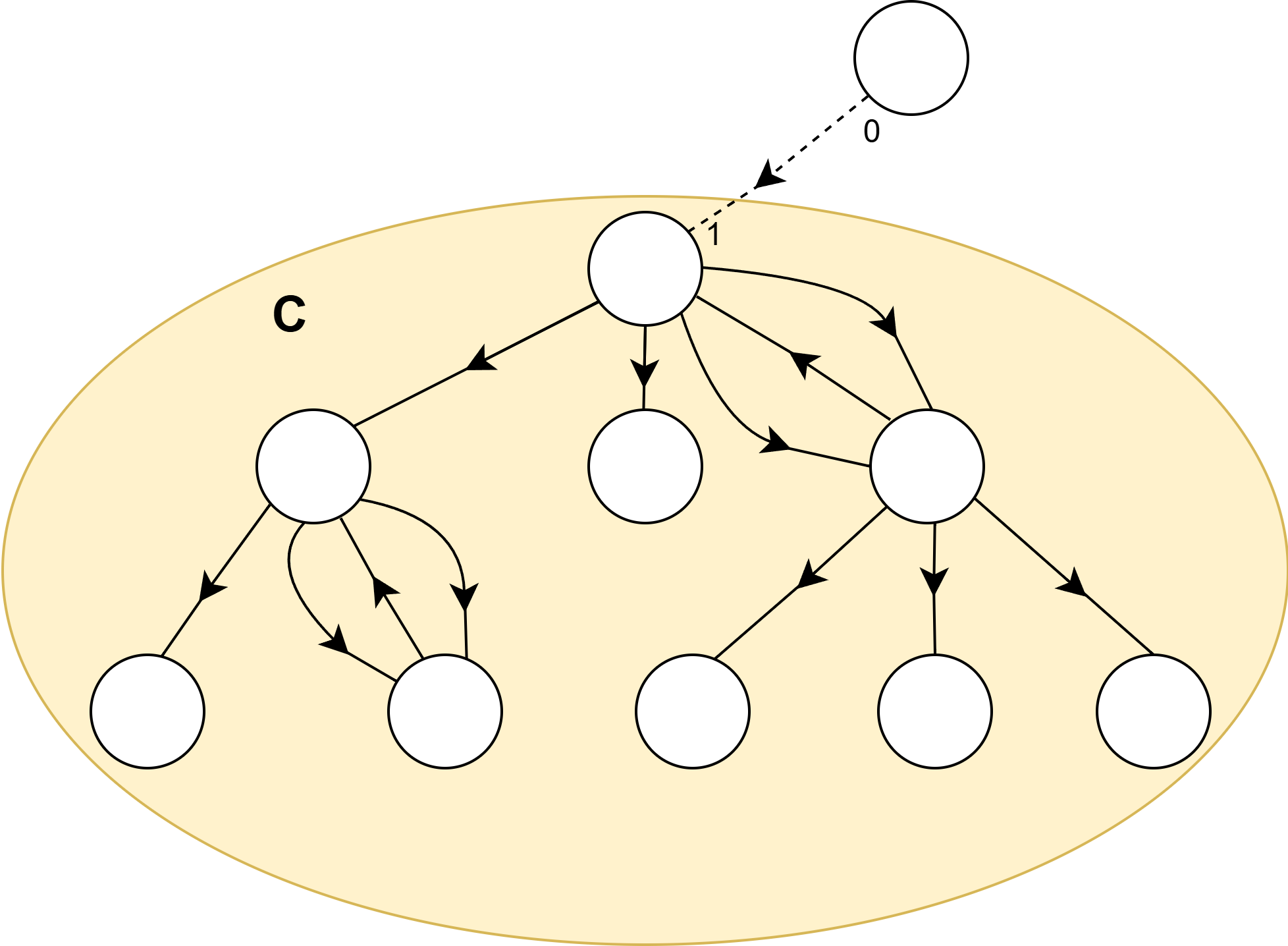}
        \caption{An example of a connected component $C$ of $G'$ that satisfies Property~\ref{prop2} of \cref{theorem:connected_component_props}, i.e., some vertex receives an edge of weight $1$ for itself and weight $0$ for the other endpoint. This allows us to repeatedly apply \cref{obs:happy_vertex} starting from that vertex, orienting all edges of $C$ to obtain an EF orientation as demonstrated.}
        \label{fig:external_edge_on_multitree_main}
    \end{figure}

\begin{figure}[ht]
        \centering
        \includegraphics[width=0.3\linewidth]{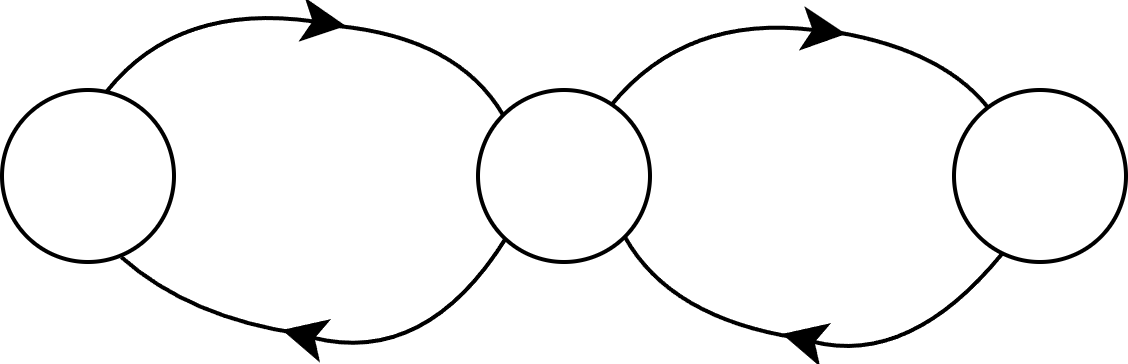}
        \caption{An example of a circuit consisting of $3$ vertices, without containing a cycle of length greater than $2$. Observe that no vertex is envious of another if the edges are oriented along a cyclic route as shown (assuming all edges have weight $1$ for both endpoints).}
        \label{fig:circuit_main}
    \end{figure}

We refer to \cref{fig:circuit_main,fig:external_edge_on_multitree_main} for some intuition for the proof of \cref{theorem:connected_component_props}.

Using \cref{theorem:connected_component_props}, we can obtain the following.

\begin{theoremrep}[\appsymb]\label{thm:EFo:linear_algo}
    There is an algorithm running in time $\bO(n+m)$ for \EFo\ with binary valuations.
\end{theoremrep}

\begin{proof}
    Define $G'=(V,E')$ as the subgraph of $G$ s.t. $E'\subseteq E$ consists only of the edges in $E$ with weight $1$ for both their endpoints. By \cref{theorem:connected_component_props}, it suffices to check whether every connected component of $G'$ satisfies at least one of the four properties mentioned in that theorem. Properties~\ref{prop2},~\ref{prop3},~\ref{prop4} can be trivially checked in $\bO(n+m)$ time for all connected components. For Property~\ref{prop1}, we observe that if a multigraph contains a circuit consisting of at least $3$ vertices and does not contain a cycle of length at least $3$, then there must exist vertices $i, j, k$ s.t. $|E'_{ij}|\geq 2$ and $|E'_{jk}|\geq 2$. Hence, for Property~\ref{prop1} it suffices to run a simple cycle-detection algorithm like depth-first search, and additionally check whether there exists a vertex $j$ with neighbors $i, k$ s.t. $|E'_{ij}|\geq 2$ and $|E'_{jk}|\geq 2$.
\end{proof}

\cref{thm:EFo:linear_algo} contrasts the complexity of \EFXo, which is known to be NP-hard even for binary valuations by Blazej et al.~\cite{sagt/BlazejGRS25}.
Additionally, it generalizes to multigraphs the algorithm of Misra and Sethia~\cite{aldt/MisraS24}, which works only for simple graphs.
We next extend this linear-time algorithm to the \EFmc\ problem.

\begin{theoremrep}[\appsymb]\label{thm:EFomc:linear_algo}
    There is an algorithm running in time $\bO(n+m)$ for \EFmc\ with binary valuations.
\end{theoremrep}

\begin{proof}
    Define $G'=(V,E')$ as the subgraph of $G$ s.t. $E'\subseteq E$ consists only of the edges in $E$ with weight $1$ for both their endpoints. By \cref{theorem:connected_component_props}, an EF orientation for $G$ exists if and only if every connected component of $G'$ satisfies at least one of the properties mentioned in that theorem. It follows that the answer to \EFmc\ is the sum over the minimum amount of edges that must be removed for each connected component to satisfy one of the four properties.

    By \cref{thm:EFo:linear_algo}, we can check whether a connected component satisfies some of the four properties in $\bO(n+m)$ time (in which case, the answer for that component is $0$). We now assume that some component violates all four properties. If it contains vertices $i,j$ with $E'_{ij}>1$, then we can remove one edge to satisfy Property~\ref{prop3}. It is clear that Properties~\ref{prop1} and~\ref{prop2} can never be satisfied by removing edges from the input; hence, if a connected component violates all four properties and contains no vertices $i,j$ with $|E'_{ij}|>1$, then our only option is to remove all edges in that component to satisfy Property~\ref{prop4}. All of the aforementioned operations can be done in time $\bO(n+m)$.
\end{proof}

\begin{remark}
    The algorithms of \cref{thm:EFo:linear_algo,thm:EFomc:linear_algo} can also find the respective orientation (if it exists) in time $\bO(n+m)$, because the proof of \cref{theorem:connected_component_props} is constructive, i.e., it orients the whole graph by repeatedly applying \cref{obs:happy_vertex}.
\end{remark}

\subsection{Hardness Results}\label{subsec:EF:hardness}

Here we present various hardness results concerning the {\EFo} problem.
As our first result, we present a lower bound based on the Strong Exponential Time Hypothesis (SETH),
which states that SAT does not admit any $(2-\varepsilon)^n n^{\bO(1)}$ algorithm, where $n$ is the number of variables.
In particular, this shows the optimality (up to polynomial factors) of the algorithm presented in \cref{thm:EFo:Heavy}.

\begin{figure}[ht]
    \centering
    \includegraphics[width=0.55\linewidth]{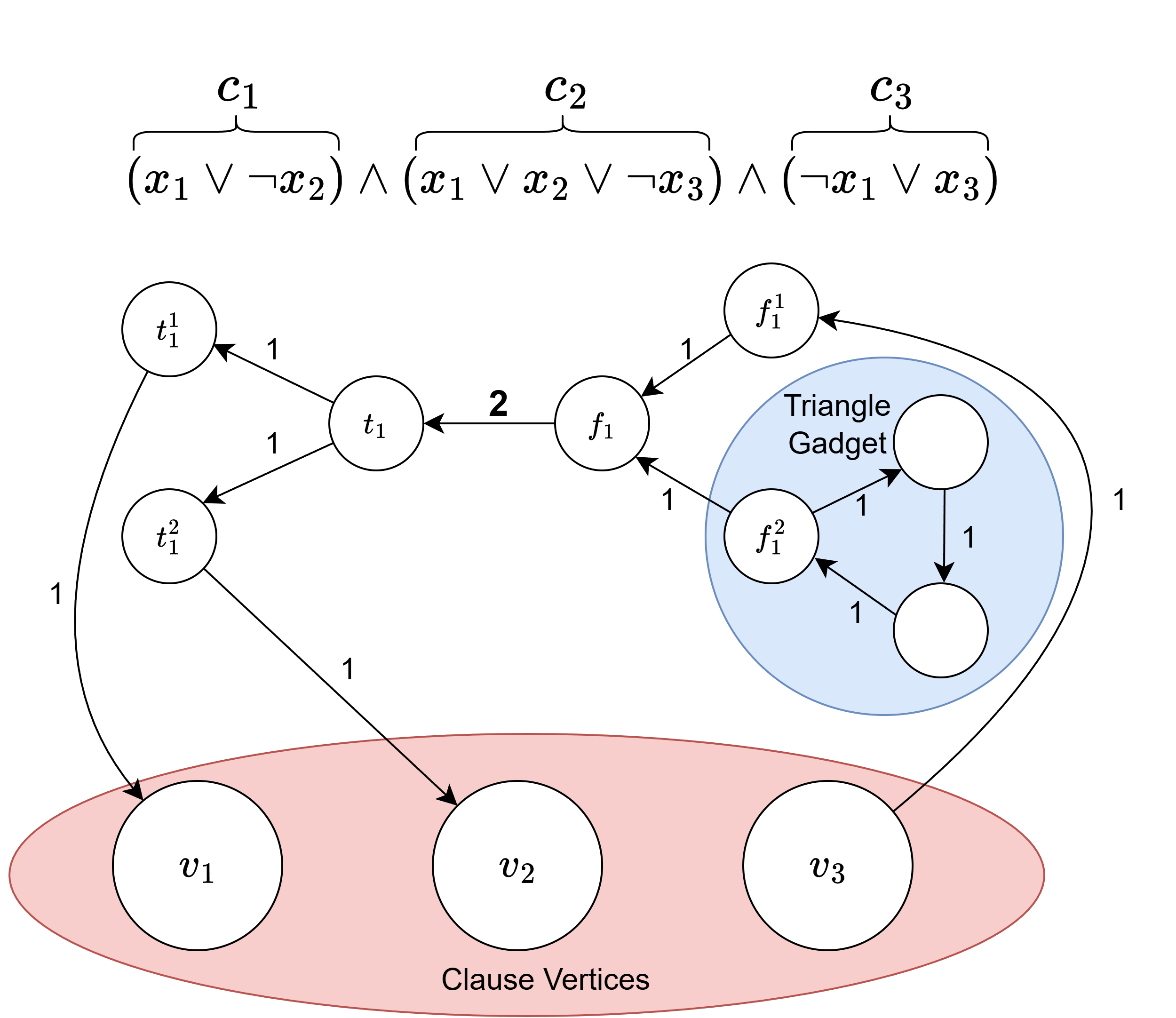}
    \caption{A sketch of our reduction from SAT to \EFo\ in \cref{thm:seth-ef}. For simplicity, we only show the part of the graph corresponding to variable $x_1$. The weighted edge $\{t_1,f_1\}$ is oriented towards $t_1$ if $x_1$ is set to $\true$, or towards $f_1$ if $x_1$ is set to $\false$. We prove that clause vertex~$v_i$ is non-envious of all its neighbors iff clause $c_i$ is satisfied. Observe that $v_1$ and $v_2$ are non-envious when $\{t_1,f_1\}$ is oriented towards $t_1$ (equivalently, $c_1$, $c_2$ are satisfied by setting $x_1$ to $\true$). For more details, see Appendix B.}
    \label{fig:sat_reduction_main}
\end{figure}

\begin{theoremrep}[\appsymb]\label{thm:seth-ef}
    There is no $(2-\varepsilon)^{k}n^{\bO(1)}$ algorithm for $\EFo$ under SETH,
    where $k$ is the number of edges of value at least $2$, even when restricted to simple symmetric instances with polynomially-bounded weights.
\end{theoremrep}

\begin{proof}
    We will reduce SAT to \EFo.
    Let $\phi$ be an instance of SAT, where $X=\{x_1,\dots,x_n\}$ denotes its variables and $C=\{c_1,\dots,c_m\}$ its clauses.
    We will construct an equivalent simple symmetric instance $\mathcal{I} = (G, v)$ of \EFo. Since the instance is symmetric, we will only keep track of one weight per edge.

    Let $w_i$ be equal to the maximum between the number of clauses of $C$ that contain $x_i$ and the number of clauses of $C$ that contain~$\neg x_i$;
    formally, $w_i = \max \{ |\setdef{c_j \in C}{x_i \in c_j}|,|\setdef{c_j \in C}{\neg x_i \in c_j}| \}$.
    For $i \in [n]$ we introduce vertices $t_i$ and~$f_i$
    and an edge $\{t_i,f_i\}$ with weight $w_i$.
    Informally, the orientation of $\{t_i,f_i\}$ towards $t_i$ corresponds to a truthful assignment of~$x_i$.
    Furthermore, for all $\ell \in [w_i]$, we introduce vertices $t_i^\ell,f_i^\ell$ and edges $\{t_i^\ell,t_i\}$ and $\{f_i^\ell,f_i\}$, all of weight $1$.

    \begin{figure}[ht]
        \centering
        \includegraphics[width=0.55\linewidth]{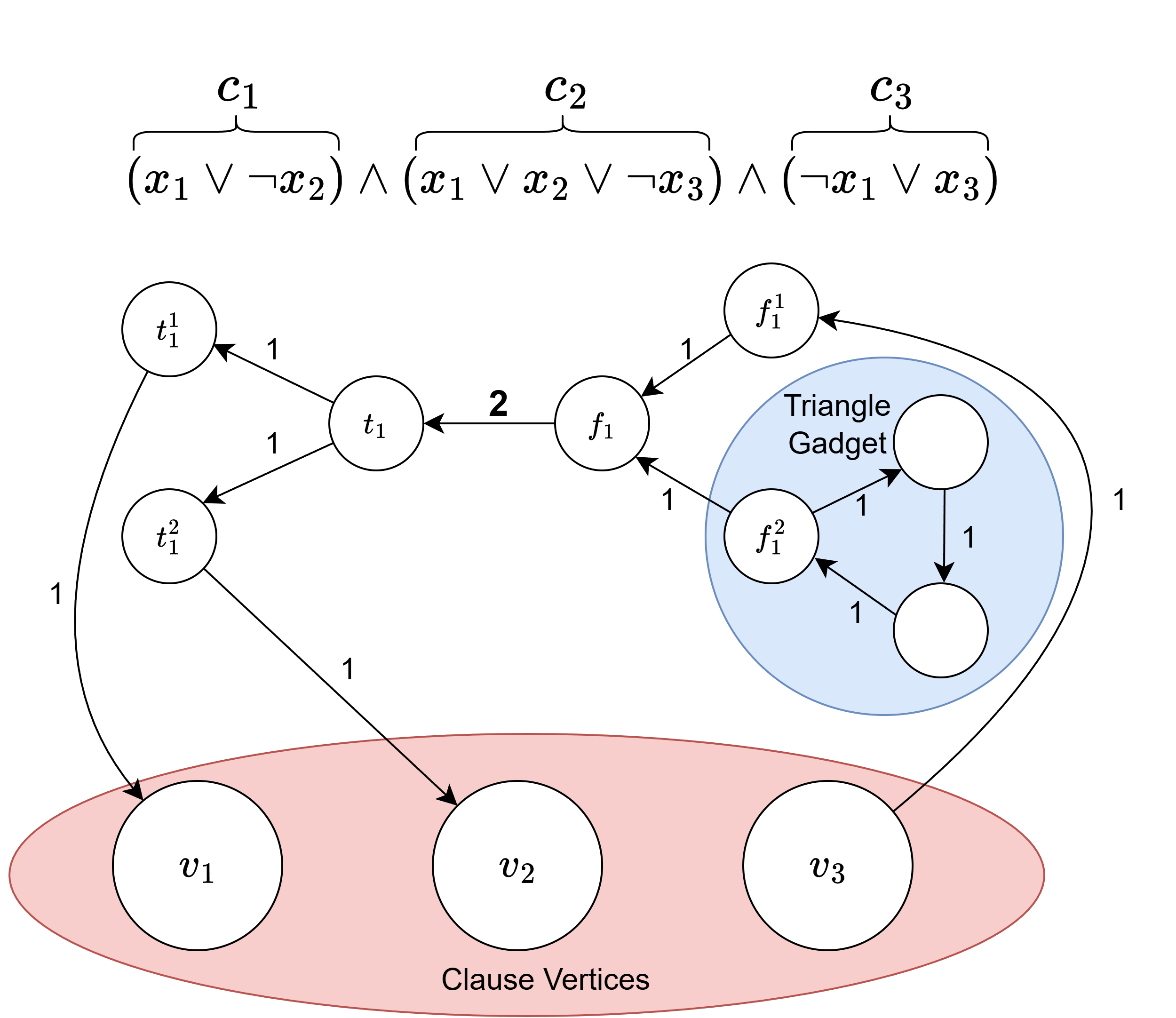}
        \caption{An example of our construction in \cref{thm:seth-ef}. For simplicity, we only show the part of the graph corresponding to $x_1$, for which $w_1=2$. Each of the two copies of $t_1$, $f_1$ is connected with some clause vertex whose respective clause contains a positive/negative occurrence of $x_1$ respectively. $f_1^2$ is \emph{leftover}, and hence attached to a triangle gadget. The orientation shown corresponds to $\alpha(x_1)=\true$. Observe that all non-clause vertices are not envious of their neighbors; additionally, all clause vertices corresponding to clauses satisfied by $\alpha(x_1)=\true$ are not envious of their neighbors, i.e., $v_3$ is the only envious vertex in the figure. Observe that, if the same construction is repeated for $x_3$ and $\alpha(x_3)=\true$, then $v_3$ will also be rendered non-envious.}
        \label{fig:sat_reduction}
    \end{figure}

    For each $c_j$ we create a vertex $v_j$ corresponding to that clause; we connect these \emph{clause} vertices with the vertices $t_i^\ell,f_i^\ell$ as follows. For $i\in [n]$, we connect $t_i^\ell$ with the vertex corresponding to the $\ell$-th clause that contains a positive occurrence of~$x_i$, with an edge of weight $1$. Note that, since the amount of clauses containing a positive occurrence of $x_i$ is at most $w_i$ (by construction), some vertices $t_i^\ell$ may remain not adjacent to any clause vertex $v_j$ by the aforementioned procedure.
    For every such \emph{leftover} vertex $u$ that is adjacent to no clause vertex, we create the following \emph{triangle} gadget. We introduce two new vertices $a_{i,\ell}^u$ and  $b_{i,\ell}^u$ and the edges $\{u,a_{i,\ell}^u\}$, $\{a_{i,\ell}^u,b_{i,\ell}^u\}$ and $\{b_{i,\ell}^u,u\}$, all of weight $1$, thus creating a triangle between $u$, $a_{i,\ell}^u$ and $b_{i,\ell}^u$.
    We apply exactly the same construction to vertices $f_i^\ell$ for \emph{negative} occurrences of $x_i$.

    This concludes the construction of $\mathcal{I}$.
    For an illustration of the construction we refer to \cref{fig:sat_reduction}. We now prove that $\phi$ has a satisfying assignment iff $\mathcal{I}$ has an EF orientation.

    ($\Rightarrow$) Let $\alpha \colon X \to \{\true, \false\}$ be a satisfying assignment for $\phi$. We construct an orientation $O$ of $G$ as follows.
    For $i \in [n]$ such that $\alpha(x_i)=\true$, we orient the edge
    $e=\{t_i,f_i\}$ towards $t_i$. Then, we orient all edges incident to $f_i$ (except~$e$) towards~$f_i$ and all edges incident to $t_i$ (except~$e$) away from $t_i$.
    Furthermore, any edge $(t_i^\ell,v_j)$ is oriented towards $v_j$, and any edge $(f_i^\ell,v_j)$ is oriented away from $v_j$.
    For $i \in [n]$ such that $\alpha(x_i)=\false$, we orient all of the aforementioned edges in the opposite manner.
    For each triangle gadget on vertices  $\{u,a_{i,\ell}^u,b_{i,\ell}^u\}$, orient its edges as a directed $3$-cycle (e.g. $u \to a_{i,\ell}^u \to b_{i,\ell}^u \to u$). Thus, each of the three vertices receives some edge of weight $1$.

    We now prove that $O$ is EF. Let $i \in [n]$ such that $\alpha(x_i)=\true$. Then, $t_i$ has received an edge of weight $w_i$ and $f_i$ has received $w_i$ edges of weight $1$ (by construction), hence $t_i, f_i$ do not envy each other. Symmetrically, the same holds for $\alpha(x_i)=\false$.
    For all vertices in the graph other than $t_i, f_i$ vertices it holds that they are only incident to edges of weight $1$. Since the graph is simple, if any vertex (other than $t_i, f_i$ vertices) receives at least one edge, then it is clear that it does not envy any of its neighbors. Observe that every vertex of the graph receives at least one edge; in particular:

    \begin{itemize}
        \item every clause vertex $v_j$ receives at least one edge because $\alpha$ is a satisfying assignment for $\phi$ (meaning that every clause $c_j$ has at least one satisfied literal, i.e., $v_j$ receives some edge from a vertex $t_i^\ell$ if the satisfied literal is positive or a vertex $f_i^\ell$ if the satisfied literal is negative).

        \item every vertex $t_i^\ell, f_i^\ell$ that is connected to some clause vertex receives either the edge connecting it to that clause vertex, or the edge connecting it to $t_i$ or $f_i$, due to the construction of $O$.

        \item every vertex $t_i^\ell, f_i^\ell$ not connected to any clause vertex $v_j$ is part of some triangle gadget, and thus receives at least one edge.
    \end{itemize}

    Thus, we have proven that no vertex envies another, meaning that $O$ is EF.

    \medskip

    ($\Leftarrow$)
    Let $O$ be an EF orientation of $G$. We define an assignment $\alpha$ by setting $\alpha(x_i)=\mathsf{true}$ iff the edge $\{t_i,f_i\}$ is oriented towards $t_i$. We will prove that $\alpha$ is a satisfying assignment.

    Let $c_j$ be a clause that is \emph{not} satisfied by $\alpha$, and let $v_j$ be its corresponding clause vertex. We will prove that $v_j$ receives no edge in $O$.
    Towards contradiction, let $e$ be an edge oriented towards $v_j$. Without loss of generality, we assume that $e = \{v_j,t^\ell_i\}$, for some $i\in [n], \ell \in [w_i]$; symmetrical arguments hold if $e = \{v_j,f^\ell_i\}$. Then, $\{t^\ell_i,t_i\}$ must be oriented towards $t^\ell_i$, otherwise $t^\ell_i$ would envy $v_j$ (recall that $t^\ell_i$ has only two incident edges, both of weight 1). Now, if $\{t_i,f_i\}$ is not oriented towards $t_i$, then $t_i$ would envy $f_i$, since only $w_i-1$ edges of weight $1$ remain available to be oriented towards $t_i$. Hence, $\{t_i,f_i\}$ must be oriented towards $t_i$; equivalently, $\alpha(x_i)=\mathsf{true}$. However, by construction, $x_i$ occurs positively in clause $c_j$ (because $v_j$ is connected with $t_i^\ell$), which implies $c_j$ is satisfied, leading to a contradiction. We infer that, for an unsatisfied clause $c_j$, its corresponding clause vertex $v_j$ must not receive any edge in $O$.

    Since $O$ is EF and there are no edges of weight $0$ in $G$, every vertex must receive at least one edge in $O$ (otherwise it would envy all of its neighbors).
    Combining this with the above, no unsatisfied clause $c_j$ exists, hence $\alpha$ is a satisfying assignment.

    \medskip

    This concludes the correctness of the reduction. Now, observe that the number of edges in $G$ of weight at least $2$ is at most $n$.
    Let $k \le n$ denote said number; by the reduction we described, a $(2-\varepsilon)^k n^{\bO(1)}$ algorithm for \EFo\ would result in a
    $(2-\varepsilon)^n n^{\bO(1)}$ algorithm for SAT, which would contradict the SETH.
\end{proof}

We refer to \cref{fig:sat_reduction_main} for a sketch of our reduction from SAT in \cref{thm:seth-ef}.
Next, by adapting essentially the same reduction from a more appropriate variant of SAT, we obtain NP-hardness for {\EFo}
even for very restricted valuation functions.

\begin{theoremrep}[\appsymb]\label{thm:npc-ef}
    $\EFo$ is NP-complete even when restricted to simple symmetric instances on $3$-regular graphs and weights in $\{1,2\}$.
\end{theoremrep}

\begin{proof}
    It is clear that \EFo\ is contained in NP, thus we only argue about its NP-hardness.
    Our reduction is roughly the same as the one presented in \cref{thm:seth-ef},
    albeit slightly simplified due to reducing from a more appropriate SAT variant.

    The problem 2P2N-3SAT is a restricted version of 3-SAT:
    an instance of 2P2N-3SAT consists of a set $X$ of $n$ variables and a set $C$ of $m$ clauses such that
    each clause has exactly three literals corresponding to three different variables
    and each variable appears exactly twice positively and exactly twice negatively.
    It is known that 2P2N-3SAT is NP-complete~\cite{eccc/ECCC-TR03-049,dam/DarmannD21}.

    For the construction, we closely follow the one of \cref{thm:seth-ef}.
    For every variable $x_i$ we introduce vertices $t_i$ and $f_i$ connected by an edge of value $2$.
    For every clause $c_j$ we introduce a clause vertex $v_j$.
    Then, if $x_i \in c_j$ (resp.~$\neg x_i \in c_j)$
    we connect $v_j$ with $t_i$ (resp.~$f_i$) and set the value of said edge to $1$.
    This completes the construction, and it is easy to see that the graph is simple.
    Since every literal appears exactly twice and every clause contains exactly $3$ literals,
    it follows that all vertices in the constructed graph are of degree $3$.

    For the correctness, we proceed as in \cref{thm:seth-ef} and we omit further details due to the
    similarity of the arguments.
\end{proof}

Moving on, we show that {\EFo} remains highly intractable even on very restricted classes of (multi)graphs.
To show this we first present a simple reduction from \textsc{Partition},
and then one from the {\TOO} problem~\cite{wg/BodlaenderCW22,dam/BodlaenderS26}.

\begin{theoremrep}[\appsymb]\label{thm:ef:weak_np_hardness}
    {\EFo} is weakly NP-complete even restricted to
    \begin{itemize}
        \item simple symmetric instances of vertex cover number $2$,
        \item symmetric instances with $2$ vertices.
    \end{itemize}
\end{theoremrep}

\begin{proof}
    It is straightforward to verify inclusion in NP.
    Our reduction follows along the lines of~\cite[Theorem 3]{ijcai/DeligkasEGK25},
    reducing from the well-known weakly NP-hard \textsc{Partition} problem which asks,
    given a multiset of positive integers, whether we can partition them into two subsets of equal subset sum.

    Let $S = \{ s_1, \ldots, s_n \}$ be a multiset of $n$ positive integers,
    such that $\sum_{i=1}^n s_i = 2B$ for some $B \in \mathbb{N}$.
    We construct a symmetric instance $\mathcal{J} = (G,v)$ of {\EFo} (equivalently,
    we construct an edge-weighted graph $G$).
    First, we introduce vertices $u_1, u_2, \ell_1, \ell_2$, and the edges $\{ u_1, \ell_1 \}$ and $\{ u_2 , \ell_2 \}$, both of weight $B$.
    Next, for all $i \in [n]$ we introduce a vertex $v_i$ and connect it with both $u_1$ and $u_2$, with both edges being of weight $s_i$.
    This concludes the construction, and it is easy to see that $G - \{u_1,u_2\}$ is an independent set, while $G$ is a simple graph.

    For the forward direction, let $S_1, S_2$ denote a partition of $S$ with $\sum_{s \in S_1} s = \sum_{s \in S_2} s = B$.
    If $v_i \in S_1$, then we orient $\{ v_i, u_1 \}$ towards $u_1$ and $\{ v_i , u_2 \}$ towards $v_i$.
    Symmetrically, if $v_i \in S_2$, then we orient $\{ v_i, u_1 \}$ towards $v_i$ and $\{ v_i , u_2 \}$ towards $u_2$.
    Lastly, for $j \in \{1,2\}$ we orient $\{ u_j, \ell_j \}$ towards $\ell_j$.
    It is straightforward to verify that this is an envy-free orientation.

    For the converse direction, let $O = \{X_v\}_{v \in V(G)}$ be an envy-free orientation of $\mathcal{J}$.
    Let $j \in \{1,2\}$ and notice that since $\ell_j$ is a leaf, the edge $\{u_j,\ell_j\}$ is oriented towards $\ell_j$ in $O$.
    Since this edge is of weight $B$ and $O$ is envy-free, it follows that edges of total weight at least $B$ are oriented towards $u_j$ in $O$.

    We argue that for all $i \in [n]$ it holds that $|X_{v_i}| = 1$.
    If $X_{v_i} = \varnothing$, $v_i$ is envious of both $u_1$ and $u_2$, a contradiction.
    Thus, for all $i \in [n]$, $|X_{v_i}| \ge 1$.
    Now assume there exists $i \in [n]$ with $|X_{v_i}| = 2$.
    In that case, it follows that the sum of the values of the edges that are oriented
    towards either $u_1$ or $u_2$ in $O$ is at most $2B - s_i$, contradicting the fact that each of $u_1$ and $u_2$ receives weight at least $B$.

    Let $S_j \subseteq S$ such that for all $v_i \in S_j$, the edge $\{ v_i, u_j \}$ is oriented towards $u_j$ in $O$.
    Due to the previous paragraph, the sets $S_1,S_2$ partition $S$, while each of $u_1,u_2$ receives goods of total value at least $B$,
    while both of them receive goods of total value exactly $2B$.
    It follows that they both receive goods of value exactly $B$, and $\sum_{s \in S_1} = \sum_{s \in S_2} = B$.

    One can easily modify the previous construction to obtain a multigraph consisting of $2$ vertices:
    we only introduce the vertices $u_1,u_2$, and for all $i \in [n]$, we introduce an edge of weight $s_i$ connecting them.
    Correctness is straightforward and we omit the details.
\end{proof}

\begin{theoremrep}[\appsymb]\label{thm:EF:w1h}
    \EFo\ is W[1]-hard parameterized by vertex cover number,
    XNLP-hard parameterized by pathwidth,
    and XALP-hard parameterized by treewidth,
    even restricted to simple symmetric instances with polynomially-bounded weights.
    Furthermore, for such instances, {\EFo} does not admit a $n^{o(\vc)}$ algorithm under the ETH,
    where $\vc$ denotes the vertex cover number of the input graph.
\end{theoremrep}

\begin{proof}
    The {\TOO} problem asks,
    given a simple graph $G=(V,E)$ with a valuation function $v \colon E \to \mathbb{N}$
    and a capacity function $c \colon V \to \mathbb{N}$,
    whether there exists an orientation $O = \{X_i\}_{i \in V}$ of $E$ such that
    for all $u \in V$ it holds that $v(X_u)=c$.

    It is known that {\TOO} is W[1]-hard parameterized by vertex cover,
    XNLP-hard by pathwidth,
    and XALP-hard by treewidth, even when all values are given in unary~\cite{wg/BodlaenderCW22,dam/BodlaenderS26}.
    Furthermore, in~\cite[Theorem~6.4]{arxiv/BodlaenderCW22} (the full version of~\cite{wg/BodlaenderCW22}),
    the authors give a reduction from {\UBP} to {\TOO} such that,
    in the produced instance, the vertex cover number equals the number of bins,
    and all weights are encoded in unary.
    It was recently shown that {\UBP} admits no algorithm running in time $n^{o(k)}$ under the ETH,
    where $k$ denotes the number of bins~\cite{bringmann2026tightsethbasedlowerbounds},
    improving upon a previous lower bound of $n^{o(k/\log k)}$~\cite{jcss/JansenKMS13};
    consequently, this results in a $n^{o(\vc)}$ ETH lower bound for {\TOO} with unary weights,
    where $\vc$ denotes the vertex cover number of the resulting graph.

    We present a reduction from {\TOO}.
    Let $\mathcal{I} =  (G,v,c)$ denote an instance of said problem.
    Assume without loss of generality that for all $\{u,v\} \in E$ it holds that $v(\{u,v\}) \le \min\{c(u),c(v)\}$,
    as otherwise either this is trivially a No instance, or the orientation of said edge would be forced and we could remove it
    by updating the capacity of its endpoints and obtain an equivalent instance.
    Furthermore, we assume that
    \begin{equation}\label{eq:w1h-ef}
        \sum_{e \in E} v(e) = \sum_{u \in V} c(u),
    \end{equation}
    as otherwise $\mathcal{I}$ is trivially a No instance.

    We construct a symmetric instance $\mathcal{J} = (G'=(V',E'),v')$ of \EFo\ such that $G$ is a simple graph.
    To obtain $G'$, we extend the graph $G$ by first adding a universal%
    \footnote{A vertex is called \emph{universal} if it is adjacent to all other vertices of the graph.}
    vertex $a$, and then attaching a leaf $b$ to $a$.
    As for the valuation function, we set $v'(e) = v(e)$ for all $e \in E$,
    $v'(\{a,u\})=c(u)$ for all $u \in V$, and $v'(\{a,b\}) = \sum_{u \in V} c(u)$.
    This completes the construction of $\mathcal{J}$,
    and it is easy to see that deleting $a$ from $G'$ removes any edges added by the construction (thus $\vc(G') \le \vc(G) + 1$),
    while the reduction requires logarithmic space.
    Furthermore, all edges in the constructed instance are of symmetric valuation,
    while the weight of any introduced edge is polynomially-bounded.

    \begin{claim}
        If $\mathcal{I}$ is a Yes instance of {\TOO},
        then $\mathcal{J}$ is a Yes instance of \EFo.
    \end{claim}

    \begin{nestedproof}
        Let $O=\{X_u\}_{u \in V}$ be an orientation of $E$ such that for all $u \in V$, $v(X_u) = c(u)$.
        We extend $O$ by orienting the edges in $E' \setminus E$ as follows to obtain an orientation~$O'$ of $E'$.
        For all $u \in V$, we orient the edge $\{a,u\}$ towards $a$, while we orient the edge $\{a,b\}$ towards $b$.

        We now argue that $O'$ is envy-free.
        Vertex $b$ receives its unique incident edge.
        Vertex $a$ receives all edges connecting it with vertices in $V$ so it is not envious of any such vertex,
        while the sum of the values of the edges it receives is exactly $\sum_{u \in V} c(u) = v'(\{a,b\})$, thus it is not envious of $b$.
        Finally, each vertex $u \in V$ receives edges of total value $c(u) = v'(\{a,u\})$ so it is not envious of $a$,
        and for all $u' \in N_G (u)$ it holds that $v'(\{u,u'\}) \le c(u)$ thus it not envious of any vertex in $V$.
    \end{nestedproof}

    \begin{claim}
        If $\mathcal{J}$ is a Yes instance of \EFo,
        then $\mathcal{I}$ is a Yes instance of {\TOO}.
    \end{claim}

    \begin{nestedproof}
        Let $O'$ be an envy-free orientation of $G'$.
        Since $b$ is a leaf, it follows that $O'$ orients the edge $\{a,b\}$ towards $b$.
        Furthermore, notice that
        \[
            \sum_{u \in V} v'(\{a,u\}) = v'(\{a,b\}),
        \]
        thus it follows that every edge incident to $a$ apart from $\{a,b\}$ is oriented towards $a$.
        Now consider $u \in V$ and notice that since
        (i) the edge $\{a,u\}$ of value $c(u)$ is oriented towards $a$,
        and (ii) $u$ is not envious of $a$,
        it follows that the sum of the values of the received edges of $u$ is at least $c(u)$,
        with all such edges belonging to $E$.
        In that case, since this holds over all $u \in V$,
        \Cref{eq:w1h-ef} implies that for every $u \in V$ the sum of the values of its received edges is \emph{exactly} $c(u)$,
        and all such edges belong to $E$.
        Consequently, it follows that the restriction of $O'$ to $E$ is an orientation $\{X_u\}_{u \in V}$ of $E$ such that for all $u \in V$, $v(X_u) = c(u)$.
    \end{nestedproof}
    This completes the proof of the theorem.
\end{proof}

\subsection{Parameterization by Number of Heavy Edges}\label{subsec:EFo_heavy}

Here we consider {\EFo} on simple graphs.
Observe that most of the hardness results of \cref{subsec:EF:hardness} already apply even under this restriction.
Our main result is an algorithm with running time $2^k n^{\bO(1)}$,
where $k$ denotes the number of edges that have value at least $2$ for at least one endpoint;
by \cref{thm:seth-ef}, this is optimal (up to polynomial factors) under the SETH.
We remark that our algorithm works even for binary-encoded weights on the edges.

The rough idea of our algorithm is the following.
As a first step we fix an orientation of all ``heavy'' edges, that is, edges of value at least $2$ for some endpoint.
This leaves a residual instance in which every unoriented edge has value~$1$ for both endpoints.
To appropriately orient these remaining edges, we reduce to an equivalent instance of
\textsc{Upper Degree-Constrained Graph Orientation}, for which a polynomial-time algorithm is known~\cite{FrankGyarfas1978}.
This allows us, for each fixed orientation of the heavy edges, to decide in polynomial time
whether it can be extended to an envy-free orientation.
Finally, iterating over all possible orientations of the heavy edges (of which there are at most $2^k$)
suffices to determine whether the instance admits an envy-free orientation.

%

\begin{theoremrep}[\appsymb]\label{thm:EFo:Heavy}
    There is an algorithm
    for \EFo\ on simple graphs running in time $2^k n^{\bO(1)}$, where $k$ is the number of edges with value at least $2$ for at least one endpoint.
\end{theoremrep}

\begin{proof}
    Let $\mathcal{I} = (G=(V,E),\{v_i\}_{i \in V})$ be an instance of \EFo, where $G$ is a simple graph.
    Given a partial orientation $O$ of $G$ and a vertex $i \in V$, we denote by
    \begin{align*}
        \indeg_O(i) &= \setdef{ e \in E }{ e \text{ is oriented towards } i \text{ in } O },\\
        \outdeg_O(i) &= \setdef{ e \in E }{ e \text{ is oriented outwards from } i \text{ in } O },\\
        \unorienteddeg_O(i) &= \setdef{ e \in E }{ e \text{ is incident with } i \text{ and is not oriented in } O }.
    \end{align*}

    By \cref{obs:zero_value_goods}, we can first preprocess the instance $\mathcal{I}$ as follows.
    We delete all edges of value~$0$ for both endpoints, while,
    if an edge has value~$0$ for one endpoint and positive value for the other,
    we orient it towards the latter.
    Finally, we branch over all the (at most) $2^k$ orientations of the unoriented edges that have value at least $2$ for at least one endpoint.
    After fixing one such branch, we obtain a partial orientation~$O$ in which
    every edge $e = \{i,j\}$ that is still unoriented satisfies $v_i(e) = v_j(e) = 1$.
    We define the \emph{revenue} of $i$ under $O$ as
    \[
        r_O(i) = \sum_{e \in \indeg_O(i)} v_i(e),
    \]
    and its \emph{demand} under $O$ as
    \[
        d_O(i) =
        \begin{cases}
            \max\bigl( \{0\} \cup \setdef{ v_i(e) }{ e \in \outdeg_O(i) } \bigr)
                & \text{if } \unorienteddeg_O(i) = \varnothing,\\[2mm]

            \max\bigl( \{1\} \cup \setdef{ v_i(e) }{ e \in \outdeg_O(i) } \bigr)
                & \text{otherwise.}
        \end{cases}
    \]
    Intuitively, $r_O(i)$ is the value allocated to agent $i$ by $O$,
    while, as \cref{claim:ef:heavy_edges} shows, $d_O(i)$ is the minimum value that agent $i$ must receive in any orientation $O^\star$ of $G$ that extends $O$
    in order not to envy any of the other agents.

    \begin{claim}\label{claim:ef:heavy_edges}
        Let $O^\star$ be an orientation of $G$ that extends $O$, and let $i \in V$.
        Agent $i$ does not envy any of its neighbors in $O^\star$ if and only if $r_{O^\star}(i) \ge d_{O}(i)$.
    \end{claim}

    \begin{nestedproof}
        Observe that $\outdeg_O(i) \subseteq \outdeg_{O^\star}(i) \subseteq \outdeg_O(i) \cup \unorienteddeg_O(i)$,
        and that for every $e \in \unorienteddeg_O(i)$ we have $v_i(e) = 1$.
        Since $G$ is a simple graph, each neighbor of $i$ shares with $i$ exactly one edge.
        Moreover, $i$ only values edges incident to it, so $i$ does not envy any other agent in $O^\star$
        if and only if $r_{O^\star}(i) \ge v_i(e)$ for all $e \in \outdeg_{O^\star}(i)$.

        If $\unorienteddeg_O(i) = \varnothing$ the statement follows straightforwardly.
        Now assume otherwise, that is, $\unorienteddeg_O(i) \neq \varnothing$.
        If $\outdeg_{O^\star}(i) \supsetneq \outdeg_O(i)$, then once again the statement follows
        straightforwardly, since there exists $e \in \outdeg_{O^\star}(i) \cap \unorienteddeg_O(i)$ such that $v_i(e) = 1$.

        It remains to consider the case where $\outdeg_{O^\star}(i) = \outdeg_O(i)$ and $\unorienteddeg_O(i) \neq \varnothing$.
        It suffices to argue that for all envy-free orientations $O^\star$, $r_{O^\star}(i) \ge 1$.
            Assume otherwise, and let $e = \{ i, j \} \in \unorienteddeg_O(i)$.
        It holds that $e$ is oriented towards $j$ in $O^\star$, while $r_{O^\star}(i) < 1 = v_i(e)$,
        thus $i$ is envious of $j$, a contradiction.
        This concludes the proof.
    \end{nestedproof}
    Given \cref{claim:ef:heavy_edges},
    our goal is to extend $O$ to an orientation $O^\star$ of $G$ such that $r_{O^\star}(i) \ge d_{O}(i)$ for every $i \in V$,
    or to verify that such an extension does not exist.
    Observe that for all extensions $O^\star$ of $O$ and all agents $i$ it holds that $r_{O^\star}(i) - r_O(i) \le |\unorienteddeg_O(i)|$,
    since $i$ has exactly $|\unorienteddeg_O(i)|$ incident unoriented edges in $O$, each of value~$1$.
    Consequently, if there exists $i \in V$ with $d_O(i) - r_O(i) > |\unorienteddeg_O(i)|$,
    we can immediately conclude that $O$ cannot be extended to an envy-free orientation and proceed with the next branch.
    In the following, assume that for all $i \in V$ it holds that $d_O(i) - r_O(i) \le |\unorienteddeg_O(i)| \le n$;
    the second inequality follows from the fact that $G$ is a simple graph.

    \begin{claim}\label{claim:ef:heavy_edges:UDCGO}
        There is an algorithm that determines in time $n^{\bO(1)}$ whether there exists an envy-free orientation $O^\star$ of $G$
        that extends $O$. If such an orientation exists, the algorithm returns it.
    \end{claim}

    \begin{nestedproof}
        Our overarching strategy is to cast our problem as an instance of \textsc{Upper Degree-Constrained Graph Orientation} (UDCGO)
        (see, e.g.,~\cite{approx/CyganK15,FrankGyarfas1978,soda/Gabow06} for this and related problems).
        In UDCGO we are given an undirected graph $G' = (V',E')$ and a function $u \colon V' \to \mathbb{N}$,
        and we are asked to determine whether there exists
        an orientation $O'$ of $G'$ such that for all $i \in V'$, $|\outdeg_{O'}(i)| \le u(i)$.

        Let $\mathcal{J}$ denote the instance of UDCGO obtained by setting $G' = (V, E')$,
        where $E' \subseteq E$ is the set of edges of $G$ that are not oriented in $O$,
        and by defining $u \colon V \to \mathbb{N}$ as follows, for all $i \in V$,
        \[
            u(i) =
            \begin{cases}
                |\unorienteddeg_O(i)|
                    & \text{if } r_O(i) \ge d_O(i),\\[2mm]
                |\unorienteddeg_O(i)| - (d_O(i) - r_O(i))
                    & \text{otherwise.}
            \end{cases}
        \]
        Intuitively, we search for an orientation $O'$ of $G'$ such that for every agent $i$ with $r_O(i) < d_O(i)$,
        \emph{at least} $d_O(i) - r_O(i)$ edges out of those in $\unorienteddeg_O(i)$ are oriented towards them.
        Equivalently, this is enforced by bounding the number of edges in $\outdeg_{O^\star}(i) \cap \unorienteddeg_O(i)$ by $u(i)$.
        Since any such edge has value~$1$, this ensures that
        $r_{O^\star}(i) \ge d_O(i)$ for all $i \in V$, where $O^\star = O \cup O'$.
        Therefore, by \cref{claim:ef:heavy_edges},
        $O$ can be extended to an envy-free orientation $O^\star$ of $G$ if and only if $\mathcal{J}$ is a yes-instance of UDCGO.
        Moreover, if $O'$ is an orientation of $G'$ witnessing that $\mathcal{J}$ is a yes-instance,
        then $O^\star = O \cup O'$ is the corresponding envy-free orientation of $G$.

        Lastly, we invoke the polynomial-time algorithm of~\cite[Theorem~1]{FrankGyarfas1978} to decide $\mathcal{J}$,
        which, in the yes-case, also returns the desired orientation $O'$.
    \end{nestedproof}
    Armed with \cref{claim:ef:heavy_edges:UDCGO} we simply iterate over all the at most $2^k$ partial orientations $O$ of the heavy edges,
    and for each such $O$ check in polynomial time whether it can be extended to an envy-free orientation of $G$.
    This yields a $2^k n^{\bO(1)}$-time algorithm, as claimed.
\end{proof}

\subsection{Parameterization by Treewidth plus Maximum Shared Weight}\label{subsec:EFo_shared_weight}

At this point, we have shown that {\EFo} remains computationally hard even when restricted to simple graphs of small vertex cover and symmetric valuations:
\Cref{thm:ef:weak_np_hardness} shows weak NP-hardness for graphs of vertex cover~$2$,
while for polynomially-bounded weights \cref{thm:EF:w1h} gives a $n^{o(\vc)}$ lower bound under the ETH.
A natural question would be to consider whether we can match the latter lower bound with an algorithm.
We answer this positively: using standard DP techniques, one can get an algorithm for {\EFo} running in time $W^{\bO(\tw)} n^{\bO(1)}$,
where $W$ denotes the maximum weight of the graph and $\tw$ its treewidth.
For polynomially-bounded weights, this algorithm matches the aforementioned lower bound (see \cref{remark:EF:tw}).

Pushing this idea further, we observe that in fact similar techniques apply even for multigraphs with non-symmetric valuations,
under the caveat that in this setting $W$ is a bound on the
sum of the weights of a set of parallel edges for one of its endpoints; observe that this is a straightforward generalization of the maximum weight considered in the previous setting.
We define this to be the \emph{maximum shared weight} of an instance.
As a matter of fact, our main result in this section is to show that we can solve the (more general) {\EFmc} problem in essentially the same running time.


\begin{theoremrep}[\appsymb]\label{thm:ef:fpt_tw}
    Given an instance $\mathcal{I}$ of \EFmc\ along with a nice tree decomposition of the input multigraph of width $\tw$,
    one can decide $\mathcal{I}$ in time $W^{\bO(\tw)} (n+m)^{\bO(1)}$,
    where $W$ denotes the maximum shared weight of $\mathcal{I}$ and $n$ the number of vertices.
\end{theoremrep}

\begin{proof}
    Let $\mathcal{I} = (G,\{v_i\}_{i\in V})$ denote the instance of {\EFmc}
    and $(T,\{B_t\}_{t\in V(T)})$ the given nice tree decomposition of $G$ of width $\tw$.

    \textbf{Notations and Definitions.}
    Before describing our algorithm we first fix some notations.
    For $t \in V(T)$, we denote by $V_t$ the set of vertices of $V$ appearing in the bags of the subtree of $T$ rooted at $t$
    excluding those of $B_t$;
    more formally, $V_t := \bigcup\setdef{B_{t'}}{t'\text{ is a descendant of $t$ in $T$}} \setminus B_t$.
    Furthermore, let $E_t:=\bigcup_{i\in V_t}E_i$
    (note that by property of a tree decomposition, the edges of $E_t$ are exactly the ones of $G[V_t]$ plus the ones between $V_t$ and $B_t$).
    Finally, let $G_t$ be the graph defined by $G_t:=(V_t\cup B_t,E_t)$. We note that $B_t$ is an independent set in $G_t$.

    Let $t\in V(T)$.

    We define $A_t$ to be the set of all possible partial orientations of $G_t$, and $A^{ef}_t$ the set of all possible partial orientations of $G_t$ where all vertices of $V_t$ are non-envious. Intuitively, those are the partial orientations of $G_t$ that can be completed into an envy-free orientation of $G$.

    \textbf{Signature definition}. Let $O=\{X_i\}_{B_t\cup V_t}\in A_t$. We define the \emph{signature} of $O$ with respect to $t$
    as a triple $\sigma=(r,d,k)$ with $r\colon B_t \to [ 0,W ]$, $d\colon B_t \to [ 0,W ]$ and $k\in [0,|E|]$ defined by:
    \begin{itemize}
        \item for all $i\in B_t$, $r(i)= \min\{ v_{i} (X_{i}),W \}$,
        \item for all $i\in B_t$, $d(i)=\max_{j\in V_t}\{v_i(X_j)\}$,
        \item $k=|E_t|-\sum_{i\in B_t\cup V_t}|X_i|$
    \end{itemize}

    Let $\Sigma_t$ denote the set of all possible signatures for node $t$, and notice that $|\Sigma_t| \le |E|(W+1)^{2(\tw+1)}$.
    We define the \emph{signature function} $\sgn_t \colon A_t \to \Sigma_t$ that attributes to each partial orientation of $G_t$ its signature with respect to $t$.

    Take $O\in A_t^{ef}$. Intuitively, $\sgn_t(O)=(r,d,k)$ has two types of information: $r$ and $d$ are the minimum amount of information needed to complete $O$ into an partial envy-free orientation of $G$: $r(i)$ represents what $i$ already \emph{received} from $O$, up to $W$ as $W$ is enough to never be envious of any other vertex, while $d(i)$ represents what $i$ \emph{demands}, or what it needs to receive to not be envious of any vertex of $V_t$. This value cannot be bigger than the maximum shared weight $W$. $k$ is the number of edges from $E_t$ missing in $O$, it is the amount of edges given to charity in $O$.

    \textbf{Signature valuation definition} We define the \emph{valuation function} $val_t \colon \Sigma_t \to \{0,1\}$ defined by: for all $\sigma\in \Sigma_t$, $val_t(\sigma)=1$ if and only if there is $O\in A_t^{ef}$ such that $\sgn_t(O)=\sigma$.
    \medskip

    Take $t$ to be the root of $T$; by definition of a nice tree decomposition, $B_t$ is an empty set, $V_t=V$, $E_t=E$ and $G_t=G$. Therefore, $A_t^{ef}$ is exactly the set of all partial envy-free orientations of $G$. So, take $k$ such that $k$ is minimal and there are $r,d$ such that $val_t((r,d,k))=1$, it is equal to the minimal charity of a partial envy-free orientation of $G$. So, if we have an efficient way to compute the values of $val_t$, we can efficiently solve \EFmc.

    \textbf{In the following, we describe the algorithm.} It follows the structure of a classical dynamic programming algorithm over a tree decomposition; for all $t\in V(T)$, it computes the valuation of its signatures from the valuation of the signatures of its children in $T$.

    Take $t\in V(T)$:

    \textbf{Leaf node}. If $t$ is a leaf node, then $B_t=\varnothing$, $V_t=\varnothing$ and $E_t=\varnothing$. Therefore, $A_t=A_t^{ef}=\{\varnothing\}$. Take $epy$ the only function from $\varnothing$ to $[O,W]$, $val_t((epy,epy,0))=1$ and for all $k\in [1,m]$, $val_t((epy,epy,k))=0$.

    This can be done in linear time trivially.
    \medskip

    \textbf{Introduce node}. If $t$ is an introduce node with a child $t'$ such that $B_t=B_{t'}\cup \{i\}$ for some $i\notin B_{t'}$, the following lemma allows us to compute the values of $val_t$ over $\Sigma_t$ dynamically:
    \begin{lemma}
    For all $\sigma=(r,d,k)\in \Sigma_t$,

    \[val_t(\sigma)=
    \begin{cases}
        0 & if\ r(i)\neq 0\ or\ d(i)\neq 0 \\
        val_{t'}((r|_{B_{t'}},d|_{B_{t'}},k)) & otherwise
    \end{cases}\]

    \end{lemma}

    Intuitively, we extend $val_{t'}$ to $B_t$ to define $val_t$.

    \begin{nestedproof}
    We note that $V_t=V_{t'}$. Furthermore, by definition of a tree decomposition, $i$ has no edge in common with $V_t$, and so no edge incident with $i$ is included in $E_t$: $E_t=E_{t'}$ and $i$ is an isolated vertex in $G_t$.

    Therefore, let $\sigma=(r,d,k)\in \Sigma_t$. If $d(i)\neq 0$ or $r(i)\neq 0$, then there is no $O\in A_t$ where $\sgn_t(O)=\sigma$.

    Otherwise, $d(i)=r(i)=0$, and as $E_t=E_{t'}$, let $O=\{X_j\}_{j\in B_t\cup V_t}\in A_t$, $X_i=\varnothing$. Let $O'=O\setminus\{X_i\}$, $O'\in A_{t'}$, and $\sgn_t(O)=\sigma$ if and only if $\sgn_{t'}(O')=(r|_{B_{t'}},d|_{B_{t'}},k)$.

    Finally, $O\in A_t^{ef}$ if and only if it makes all the vertices of $V_t$ non-envious, if and only if $O'$ does too, but as $V_t=V_{t'}$, it holds if and only if $O'\in A_t^{ef}$. Therefore, the lemma holds.
    \end{nestedproof}

    Concerning the complexity, supposing the values of $val_{t'}$ have already been computed, the calculation of each value of $val_t$ can be done in constant time; therefore the calculation of all the values can be done in time $|\Sigma_t|\leq m(W+1)^{2(\tw +1)}$.
    \medskip

    \textbf{Forget node}.

    Let $i,j \in V$.
    By \cref{obs:zero_value_goods} there is no edge $e \in E_{ij}$ with $v_i(e)=v_j(e)=0$,
    implying that $|E_{i,j}| \leq 2W$.
    We define the function $M_{ij} \colon [ 0,v_i(E_{ij})]^2 \times [0,v_j(E_{ij})]^2\times[0,2W] \to \{0,1\}$
    so that $M_{ij}((a_1,a_2),(b_1,b_2),k)=1$ if there is a partial orientation $O=\{X_i,X_j\}$ of the subgraph $(\{i,j\},E_{ij})$ such that $v_i(X_i)=a_1$, $v_i(X_j)=a_2$, $v_j(X_j)=b_1$, $v_j(X_i)=b_2$, and $|E_{ij}\setminus(X_i\cup X_j)|=k$. We can compute the values of $M_{ij}$ via simple dynamic programming. 

    If $t$ is a forget node with a child $t'$ such that $B_{t'}=B_t\cup \{i\}$ for some $i\notin B_t$, the following lemma allows us to compute the values of $val_t$ over $\Sigma_t$ dynamically:

    \begin{lemma}

    For all $\sigma=(r,d,k)\in\Sigma_t$:
    $val_t((r,d,k))=1$ if and only if there are $(r',d',k')\in \Sigma_{t'}$, $\{(a_j^1,a_j^2,b_j^1,b_j^2,k_j)\}_{j\in B_t}\in ([O,W]^4\times[0,2W])^{B_t}$  verifying:  \begin{itemize}
        \item $val_{t'}((r',d',k'))=1$
        \item for all $j\in B_t$, $M_{ij}((a_j^1,a_j^2),(b_j^1,b_j^2),k_j)=1$
        \item for all $j\in B_t$, $r(j)=\min\{r'(j)+b_j^1,W\}$
        \item for all $j\in B_t$, $d(j)=\max\{d'(j),b_j^2\}$
        \item $d'(i)\leq\sum_{j\in B_t}a_j^1+r'(i)$
        \item for all $j\in B_t$, $a_j^2\leq r'(i)+\sum_{j\in B_t}a_i^1$
        \item $k=k'+\sum_{j\in B_t}k_j$
    \end{itemize}

    \end{lemma}

    Intuitively, a partial orientation of $G_t$ is composed of a partial orientation of $G_{t'}$ plus a partial orientation of the edges between $i$ and $B_t$, which allows us to use the values of $val_{t'}$ to calculate the values of $val_t$.

    \begin{nestedproof}

    We note that $V_t=V_t'\cup \{i\}$, and $E_t=E_{t'}\cup \bigcup_{j\in B_t} E_{ij}$. Also, $V_t\cup B_t=V_{t'}\cup B_{t'}$.

    In that case, a partial orientation $O\in A_t$ can be seen as a partial orientation $O'\in A_{t'}$ together with a partial orientation $O''$ of $(B_t',\bigcup_{j\in B_t} E_{i,j})$. In the following, we show that given a partial orientation $O\in A_t$, its signature $\sgn_t(O)$ can be computed using only the following information: the signature of $O'$ the restriction of $O$ to $G_{t'}$, and the orientation of $E_{i,j}$ for all $j\in V_t$ of $O$.

    Let $O=\{X_j\}_{j\in V_t\cup B_t}\in A_t$.

    For all $j\in B_t$, let $\alpha_j:=E_{ij}\cap X_i$, $\beta_j:=E_{ij}\cap X_j$, and $X'_j:=X_j\setminus E_{ij}$: $\alpha_j$ is the set edges of $E_{ij}$ received by $i$ while $\beta_j$ is the set of the ones received by $j$. $X'_j$ is the set of edges received by $j$ from $G_{t'}$.

    Let $X'_i:=X_i\setminus (\bigcup_{j\in B_t}E_{ij})$ and for all $\ell\in V_{t'}$, let $X'_\ell:=X_\ell$.  $O':= \{X'_j\}_{j\in B_{t'}\cup V_{t'}}$ is a partial orientation of $G_{t'}$, the restriction of $O$ to $G_{t'}$, so $O'\in A_{t'}$. Let $\sigma=(r,d,k)=\sgn_t(O)$ and $\sigma'=(r',d',k')=\sgn_{t'}(O')$.

    We note that $k=k'+\sum_{j\in B_t}|E_{ij}\setminus(\alpha_j\cup\beta_j)|$.

    For all $j\in B_t$, let $a^j_1:= v_i(\alpha_j)$, $a^j_2:= v_i(\beta_j)$, $b^j_1:= v_j(\alpha_j)$ and $b^j_2:= v_j(\beta_j)$. By linearity of the value functions, we have: $v_i(X_i)=v_i(X'_i)+\sum_{j\in B_t}a_j^1$, and for all $j\in B_t$, $v_j(X_j)=v_j(X'_j)+b^j_1$, $v_i(X_j)=v_i(X'_j)+a_2^j=a_2^j$ and $v_j(X_i)=v_j(X'_i)+b_2^j=b_2^j$.

    By definition:\begin{itemize}
        \item for all $j\in B_t$, $r(j)=\min\{v_{j}(X_{j}),W\}$
        \item for all $j\in B_t$, $d(j)=\max_{l\in V_t}\{v_j(X_l)\}$
    \end{itemize}

    In the following, we will reformulate those two previous equalities:

    First, for all $j\in B_t$, $r(j)=\min\{v_{j}(X_{j}),W\}=\min\{v_{j}(X'_{j})+b_j^1,W\}=\min\{\min\{v_j(X'_j),W\}+b_j^1,W\}=\min\{r'(j)+b_j^1,W\}$.

    Second, for all $j\in B_t$, $\max_{l\in V_t}\{v_j(X_l)\}=\max\{\max_{\ell\in V_t'}\{v_j(X'_l)\},b_2^j\}=\max\{d'(j),b_2^j\}$.

    We obtain the two following formulas: for all $j\in B_t$, $r(j)=\min\{r'(j)+b_j^1,W\}$ and $d(j)=\max\{d'(j),b_2^j\}$.

    Furthermore, $O\in A_t^{ef}$ if , by definition, $O$ makes the vertices of $V_t$ non-envious, if and only if it does on $V_t'$ and on $\{i\}$.

    First, for all $\ell\in V_t'$, $\ell'\in V_t\cup B_t$, $E_{\ell\ell'}\subseteq E_{t'}$. So, as $O'$ is the restriction of $O$ on $E_{t'}$, $O$ makes the vertices of $V_t'$ non-envious if and only if $O'$ does, and therefore if and only if $O'\in A_{t'}^{ef}$.

    Second, $i$ is non-envious if it is non-envious of the vertices of $B_t$ and $V_{t'}$. But, on one side, for all $j\in B_t$, $v_i(X_i)=v_i(X'_i)+\sum_{j\in B_t}a_i^1$ and $v_i(X_j)=a_2^j\leq W$, so $v_i(X_i)\geq v_i(X_j)$ if and only if $r'(i)+\sum_{j\in B_t}a_i^1\geq a_2^j$,
    and on the other side, $\max_{\ell\in V_{t'}}\{v_i(X_\ell)\}=\max_{\ell\in V_{t'}}\{v_i(X'_\ell)\}=d'(i)$ and $v_i(X_i)= v_i(X'_i)+\sum_{j\in B_t}a_j^1$, so as $d'(i)\leq W$,  $v_i(X_i)\geq \max_{\ell\in V_{t'}}\{v_i(X_\ell)\}$ if and only if $r'(i)+\sum_{j\in B_t}a_j^1\geq d'(i)$.

    So, $O$ is envy-free on $V_t$ if and only if $O'$ is envy-free on $V_t'$, $d'(i)\leq r'(i)+\sum_{j\in B_t}a_j^1$ and for all $j\in B_t$, $r'(i)+\sum_{j\in B_t}a_i^1\geq a_2^j$.

    We can deduce the following property: let $\sigma=(r,d,k)$ be a signature of $t$. Then $val_t(\sigma)=1$ if and only if there is $\sigma'=(r',d',k')$ a signature of $t'$, $O''=\{X''_j\}_{j\in B_t'}$ an orientation of $\{B_t',\bigcup_{j\in B_t}E_{ij}\}$ respecting:
    \begin{itemize}
        \item $val_{t'}(\sigma')=1$
        \item for all $j\in B_t$, $r(j)=min\{W,r'(j)+v_j(X''_j)\}$
        \item for all $j\in B_t$, $d(j)= \max\{d'(j),v_j(X''_i)\}$
        \item $d'(i)\leq v_i(X''_i)+v_i(X'_i)$
        \item for all $j\in B_t$, $r'(i)+v_i(X''_i)\geq v_j(X''_i)$
        \item $k=k'+ |E_{ij}\setminus (\bigcup_{j\in B_{t'}}X''_j)|$
    \end{itemize}

    Therefore, we can deduce the following recursive formula:

    $val_t((r,d,k))=1$ if and only if there are $(r',d',k')\in \Sigma_{t'}$, $\{(a_j^1,a_j^2,b_j^1,b_j^2,k_j)\}_{j\in B_t}\in ([O,W]^4\times[0,2W])^{B_t}$  verifying: \begin{itemize}
        \item $val_{t'}((r',d',k'))=1$
        \item for all $j\in B_t$, $M_{ij}((a_j^1,a_j^2),(b_j^1,b_j^2),k_j)=1$
        \item for all $j\in B_t$, $r(j)=\min\{r'(j)+b_j^1,W\}$
        \item for all $j\in B_t$, $d(j)=\max\{d'(j),b_j^2\}$
        \item $d'(i)\leq\sum_{j\in B_t}a_j^1+r'(i)$
        \item for all $j\in B_t$, $a_j^2\leq(r'(i)+\sum_{j\in B_t}a_i^1)$
        \item $k=k'+\sum_{j\in B_t}k_j$
    \end{itemize}

    \end{nestedproof}

    Concerning the complexity, supposing the values of $val_{t'}$ have already been computed, there are at most $(W^5)^{\bO(tw)}$ different set quadruplets of integers and $m\times W^{\bO(\tw)}$ different signatures of $t'$ to consider. Therefore, this step can be done in time $m\times W^{\bO(tw)}\times W^{\bO(tw)}=m\times W^{\bO(tw)}$.
    \medskip

    \textbf{Join node}. If $t$ is a join node with two children $t_1,t_2$ such that $B_t=B_{t_1}=B_{t_2}$, the following lemma allows us to compute the values of $val_t$ over $\Sigma_t$ dynamically:

    \begin{lemma}
    For all $\sigma=(r,d,k)\in\Sigma_t$:
    $val_t((r,d,k))=1$ if and only if there are
     $\sigma_1=(r_1,d_1,k_1)\in \Sigma_{t_1}$ and $\sigma_2=(r_2,d_2,k_2)\in \Sigma_{t_2}$ such that:
    \begin{itemize}
        \item $val_{t_1}(\sigma_1)=1$ and $val_{t_2}(\sigma_2)=1$
        \item for all $j\in B_t$, $r(j)=\min\{r_1(j)+r_2(j),W\}$
        \item for all $j\in B_t$, $d(j)=\max\{d_1(j),d_2(j)\}$
        \item $k=k_1+k_2$
    \end{itemize}

    \end{lemma}

    Intuitively, a partial orientation of $G_t$ is a combination of a partial orientation of $G_{t_1}$ and $G_{t_2}$; therefore, we can combine any possible valid partial orientation of $G_{t_1}$ and $G_{t_2}$ into a valid partial orientation of $G_t$.
    \begin{nestedproof}
    Note that $V_t=V_{t_1}\cup V_{t_2}$, with $V_{t_1}\cap V_{t_2}=\varnothing$. Similarly, $E_t=E_{t_1}\cup E_{t_2}$, with $E_{t_1}\cap E_{t_2}=\varnothing$.

    Let $O=\{X_j\}_{j\in\{B_t\cup V_t\}}\in A_t$; for all $j\in B_{t_1}\cup V_{t_1}$, let $X_j^1:=X_j\cap E_{t_1}$ and
    for all $j\in B_{t_2}\cup V_{t_2}$, let $X_j^2:=X_j\cap E_{t_2}$.

    Let $O_1:=\{X_j\}_{j\in\{B_{t_1}\cup V_{t_1}\}} $ and $O_2:=\{X_j\}_{j\in\{B_{t_2}\cup V_{t_2}\}}$. They are the restrictions of $O$ to $G_{t_1}$ and $G_{t_2}$ respectively: $O_1\in A_{t_1}$ and $O_2\in A_{t_2}$.

    Let $\sigma=(r,d,k)=\sgn_t(O)$, $\sigma_1=(r_1,d_1,k_1)=\sgn_{t_1}(O_1)$ and $\sigma_2=(r_2,d_2,k_2)=\sgn_{t_2}(O_2)$.

    for all $j\in B_t$, $r(j)=\min\{v_j(X_j),W\}$ and $d(j)=\max_{l\in V_t}\{v_j(X_l)\}$ by definition. In the following, we will rewrite those two equations:
    \begin{align*}
        r(j) & = \min\{v_j(X_j),W\}\\ &=\min \{v_j(X_j^1)+v_j(X_j^2),W\}\\ &=\min\{r_1(j)+r_2(j),W\}
    \end{align*}
    \begin{align*}
        d(j)&=\max_{l\in V_t}\{v_j(X_l)\}\\&=\max\{\max_{l\in V_{t_1}}\{v_j(X_l)\},\max_{l\in V_{t_2}}\{v_j(X_l)\}\}\\&=\max\{d_1(j),d_2(j)\}
    \end{align*}

    In the end, we obtain $r(j)=\min\{r_1(j)+r_2(j),W\}$ and $d(j)=\max\{d_1(j),d_2(j)\}$.

    Furthermore, $O$ is envy-free on $V_t$ if and only if it is envy-free on $V_{t_1}$ and $V_{t_2}$, if and only if $O_1$ is envy-free on $V_{t_1}$ and $O_2$ is envy-free on $V_{t_2}$. So, $O\in A_t^{ef}$ if and only if $O_1\in A_{t_1}^{ef}$ and $A_{t_2}^{ef}$.

    We obtain the following recursive formula:

    $val_t((r,d,k))=1$ if and only if there are
     $\sigma_1=(r_1,d_1,k_1)\in \Sigma_{t_1}$ and $\sigma_2=(r_2,d_2,k_2)\in \Sigma_{t_2}$ such that:
    \begin{itemize}
        \item $val_{t_1}(\sigma_1)=1$ and $val_{t_2}(\sigma_2)=1$
        \item for all $j\in B_t$, $r(j)=\min(r_1(j)+r_2(j),W)$
        \item for all $j\in B_t$, $d(j)=\max\{d_1(j),d_2(j)\}$
        \item $k=k_1+k_2$
    \end{itemize}

    \end{nestedproof}

    Concerning the complexity, supposing the values of $val_{t_1}$ and $val_{t_2}$ have already been computed, there are $m\times W^{\bO(tw)}$ signatures of $t_1$ and $t_2$ to consider. Therefore, this step can be done in time $m W^{\bO(tw)}\times mW^{\bO(tw)}=m^2 W^{\bO(tw)}$.
    \medskip

    This concludes the description and correctness of our algorithm. We now argue about its complexity: as the nice tree decomposition can be supposed to be of polynomial size and each step can be done in time at most $m^2 W^{\bO(tw)}$, the algorithm has complexity $W^{\bO(tw)}(nm)^{\bO(1)}$.
\end{proof}

\begin{proofsketch}
    Our algorithm is a standard DP over a nice tree decomposition of the input graph.
    For a node $t$ of the given nice tree decomposition, let $G_t$ denote the corresponding subgraph,
    while $B_t$ denotes the set of vertices appearing in the bag of node $t$.
    For every node $t$ of the tree decomposition, we keep track of whether there exists a partial orientation of $G_t$
    (excluding the edges among vertices of $B_t$)
    that guarantees that the vertices of $V(G_t) \setminus B_t$ do not envy any vertex of $V(G_t)$,
    while also keeping track of the number of unoriented edges of the partial orientation
    as well as the following information for each of the at most $\tw+1$ vertices in $B_t$:
    \begin{itemize}
        \item its \emph{revenue} due to edges shared with vertices in $V(G_t) \setminus B_t$,
        \item its \emph{demand} so that it does not envy any vertex of $V(G_t) \setminus B_t$.
    \end{itemize}
    Observe that both the revenue and the demand of a vertex in $B_t$ can be upper-bounded by~$W$,
    as a vertex of revenue $W$ cannot envy any other vertex, while a vertex may envy another only up to a value of at most $W$.
    Lastly, by \cref{obs:zero_value_goods} we can assume that there are no edges of valuation~$0$ for both endpoints,
    thus there are at most $2W$ edges shared between any two vertices, implying that
    the number of unoriented edges is at most $2W \cdot \binom{n}{2}$.
    Consequently, the total running time is $W^{\bO(\tw)} (n+m)^{\bO(1)}$.
\end{proofsketch}

We obtain the following corollary for polynomially bounded weights, in which case we have $W = (n+m)^{\bO(1)}$.%

\begin{corollary}\label{cor:EFmc_poly_weights_algo}
    There is an algorithm running in time $(n+m)^{\bO(\tw)}$ for \EFmc\ with polynomially-bounded weights.
\end{corollary}

\begin{remark}\label{remark:EF:tw}
    The algorithm of \cref{cor:EFmc_poly_weights_algo} matches the lower bound of $n^{o(\vc)}$ in \cref{thm:EF:w1h}.
    Recall that the treewidth of a graph is upper-bounded by its vertex cover number and that for simple graphs, $m = \bO(n^2)$ holds.
\end{remark}

\section{EFX Orientations}

The following observations will be useful throughout this section. Recall the definition of \emph{strong} envy from \cref{def:envy}.

\begin{observation}\label{obs:EFX_one_good}
    An agent that has only received one good is not strongly envied by any agent.
\end{observation}

\begin{observation}\label{obs:EFX_zero_value}
    For agents $i,j$, if $i$ does not strongly envy~$j$ and $j$ has received a good $g$ s.t.~$v_i(g) = 0$, then $i$ does not envy $j$.
\end{observation}

\subsection{From EF to EFX Orientations}\label{subsec:EF->EFX}

It was first proven by Christodoulou et al.~\cite{EFXgraphs} that \EFXo\ is NP-complete even with symmetric valuations;
this hardness persists even for simple graphs of vertex cover number $8$ or for multigraphs consisting of $10$ vertices,
where the problem was proven to be weakly NP-complete (Deligkas et al.~\cite{ijcai/DeligkasEGK25}).
In a follow-up work, Afshinmehr et al.~\cite{AfshinmehrDKMR25} show NP-completeness even for symmetric instances on bipartite multigraphs with $8$ vertices.
Here we present a reduction from \EFo\ to \EFXo\ that, along with \cref{subsec:EF:hardness},
will allow us to complement and improve upon the aforementioned results. See \cref{fig:EF_to_EFX_main} for an overview of our reduction.

   \begin{figure}[ht]
        \centering
        \includegraphics[width=0.5\linewidth]{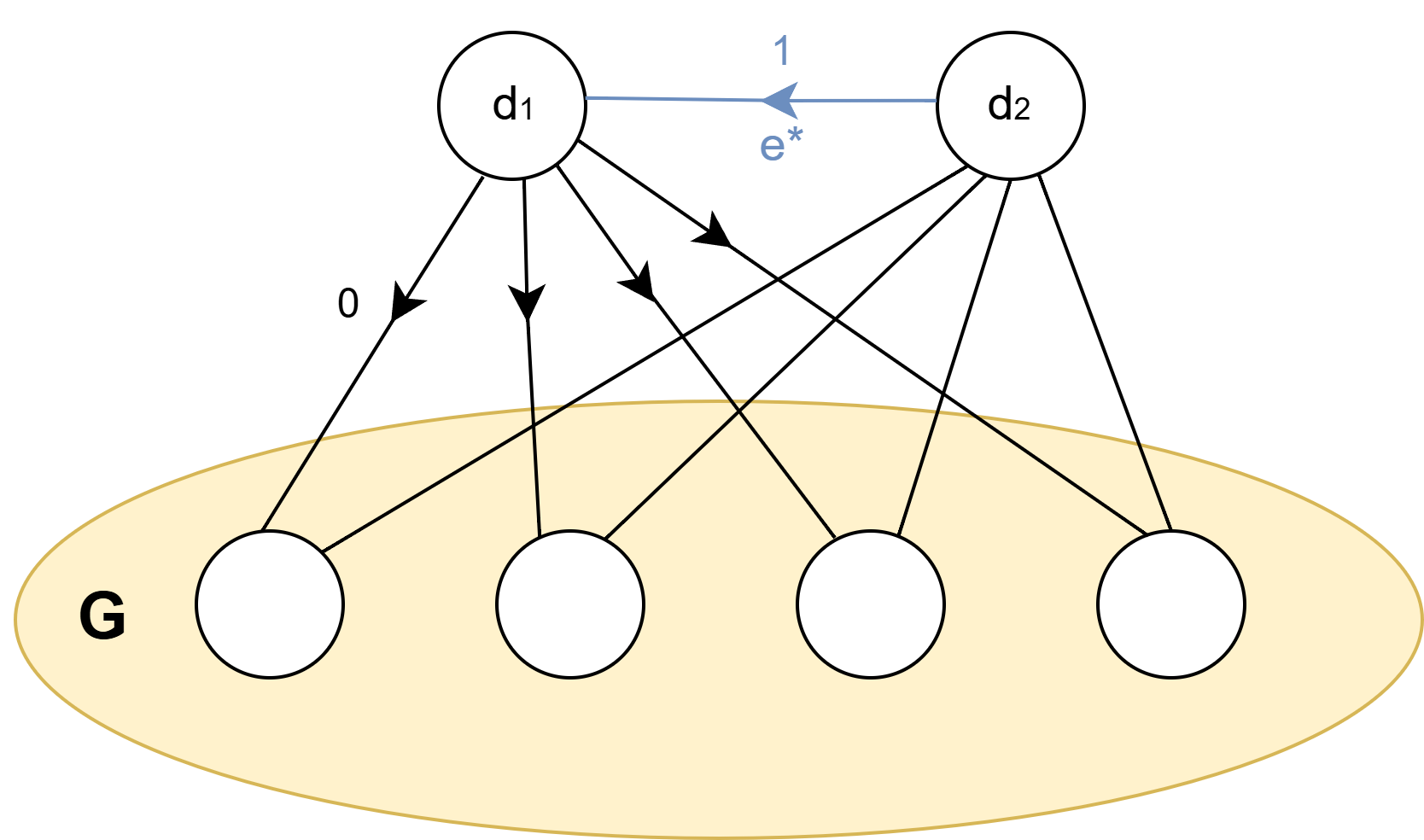}
        \caption{Our reduction from \EFo\ to \EFXo\ in \cref{thm:efo->efxo}. If $e^*$ is oriented towards~$d_1$, then all other edges incident to $d_1$ must be oriented away from $d_1$ to obtain an EFX orientation. Every $v\in G$ receives an edge of weight $0$, which forces equivalence between EF and EFX by \cref{obs:EFX_zero_value}. A symmetrical argument holds if $e^*$ is oriented towards $d_2$. For more details, see Appendix C.}
        \label{fig:EF_to_EFX_main}
    \end{figure}

\begin{theoremrep}[\appsymb]\label{thm:efo->efxo}
    There is a polynomial-time, logarithmic-space reduction from \EFo\ to \EFXo\ that only adds 2 vertices
    and edges of symmetric binary valuation.
    The reduction preserves simplicity.
\end{theoremrep}

\begin{proof}
    Let $\mathcal{I}=(G=(V,E),\{v_i\}_{i\in V})$ be an instance of \EFo.
    We build an $\EFXo$ instance $\mathcal{I'}$ by adding two vertices $d_1,d_2$ to $G$, each connected with every $v\in V$ by edges of weight $0$ (for both endpoints). Additionally, we connect $d_1$ and $d_2$ with an edge $e^*$ of weight $1$ (for both endpoints). Let $G'=(V',E')$ be the resulting graph of~$\mathcal{I}'$ (see \cref{fig:EF_to_EFX}).

    \begin{figure}[ht]
        \centering
        \includegraphics[width=0.5\linewidth]{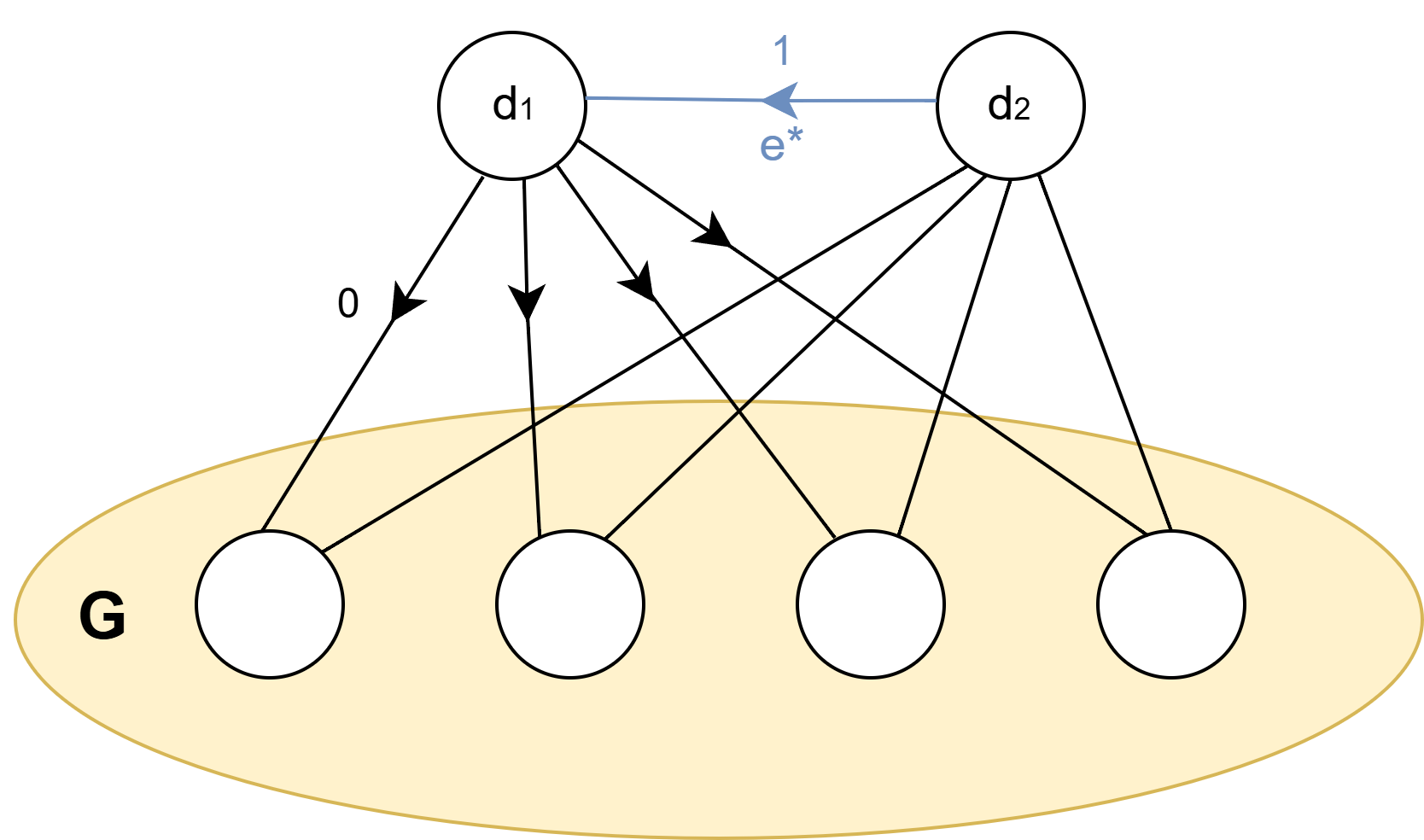}
        \caption{Our reduction from \EFo\ to \EFXo. If $e^*$ is oriented towards~$d_1$, then all other edges incident to $d_1$ must be oriented away from $d_1$ to obtain an EFX orientation. Every $v\in G$ receives an edge of weight $0$, which forces equivalence between EF and EFX by \cref{obs:EFX_zero_value}. A symmetrical argument holds if $e^*$ is oriented towards $d_2$.}
        \label{fig:EF_to_EFX}
    \end{figure}

    We prove that if $\mathcal{I}$ has an EF orientation, then $\mathcal{I}'$ has an EFX orientation. Take an EF orientation $O$ of $\mathcal{I}$ and extend it by orienting the edges in $E'\setminus E$ as follows to obtain an orientation~$O'$.
    Orient $e^*$ towards $d_1$. Then, orient all other edges incident to $d_2$ towards $d_2$ and all other edges incident to $d_1$ away from $d_1$. By \cref{obs:EFX_one_good}, $d_1$ is not strongly envied by any $v\in V'$; the same holds for $d_2$ since all of its allocated edges are of weight $0$ (for both endpoints). Additionally, $d_1$ does not envy any other vertex, since
    it has received its unique incident edge with non-zero weight. Observe that $d_2$ only envies $d_1$, since it has received all of its incident edges except $e^*$ (however, this envy is not strong, by \cref{obs:EFX_one_good}). Since $O$ is EF, it holds that $u$ does not envy $v$ for all $u,v \in V$, implying that $u$ does not strongly envy $v$. From all the above we infer that $O'$ is EFX.

    We now prove the converse: if $\mathcal{I}'$ has an EFX orientation, then $\mathcal{I}$ has an EF orientation. Take an EFX orientation $O'$ of $\mathcal{I}'$. Without loss of generality, assume $e^*$ is oriented towards $d_1$ in $O'$. Then, for $O'$ to be an EFX orientation, $d_1$ must not receive any other edge (otherwise $d_2$ would strongly envy $d_1$), implying that all other edges incident to $d_1$ are oriented away from $d_1$ in $O'$. Hence, every vertex $v\in V$ has received a zero-weight edge in $O'$. For all $u,v\in V$, $u$ does not strongly envy $v$ because $O'$ is EFX; by \cref{obs:EFX_zero_value}, this implies that $u$ does not envy $v$. Hence, orienting all edges in $E$ in accordance to $O'$ results in an EF orientation $O$ of $\mathcal{I}$.

    The reduction described uses logarithmic space,%
    \footnote{This property is essential for the XNLP- and XALP-hardness in \cref{cor:EFXo_W1}.}
    since only two vertices and $\bO(n)$ edges of polynomial weight are added to the graph.
\end{proof}

Combining \cref{thm:efo->efxo} with the hardness results in \cref{subsec:EF:hardness} leads to \cref{cor:EFXo:weak_NPh,cor:EFXo_W1}.

\begin{corollary}\label{cor:EFXo:weak_NPh}
    {\EFXo} is weakly NP-complete even restricted to
    \begin{itemize}
        \item simple symmetric instances of vertex cover number $4$,
        \item symmetric instances with $4$ vertices.
    \end{itemize}
\end{corollary}

\begin{corollary}\label{cor:EFXo_W1}
    \EFXo\ is  W[1]-hard parameterized by vertex cover number,
    XNLP-hard parameterized by pathwidth,
    and XALP-hard parameterized by treewidth,
    even restricted to simple symmetric instances with polynomially-bounded weights.
    Furthermore, for such instances, {\EFXo} does not admit a $n^{o(\vc)}$ algorithm under the ETH,
    where $\vc$ denotes the vertex cover number of the input graph.
\end{corollary}

We now compare \cref{cor:EFXo:weak_NPh,cor:EFXo_W1} with existing literature.
\cref{cor:EFXo:weak_NPh} improves upon the respective constants of 8 and 10 from Deligkas et al.~\cite{ijcai/DeligkasEGK25}, as well as the result for multigraphs with~8 vertices by Afshinmehr et al.~\cite{AfshinmehrDKMR25}.
As for \cref{cor:EFXo_W1}, it presents the first hardness results for instances with polynomially-bounded weights and small vertex cover,
and along with the algorithm of \cref{subsec:EFXo:tw}, it answers an open question of Deligkas et al.~\cite{ijcai/DeligkasEGK25} regarding the hardness of {\EFXo} in this setting.



\subsection{Parameterization by Treewidth plus Maximum Shared Weight}\label{subsec:EFXo:tw}

Moving on, we consider the {\EFXo} problem on graphs of small treewidth.
As previously mentioned, if one allows binary-encoded weights, then the problem remains
NP-hard even for simple graphs of constant vertex cover, or multigraphs with constant number
of vertices. Consequently, in the following we consider the case where all weights are polynomially-bounded;
this is exactly the regime for which Deligkas et al.~\cite{ijcai/DeligkasEGK25} have explicitly asked whether
the structure of the input graph may prove algorithmically useful.
For this case, \cref{cor:EFXo_W1} already gives a $n^{o(\vc)}$ lower bound under the ETH even on
simple graphs. Our main contribution in this section is to complement this lower bound with a DP algorithm parameterized by treewidth
plus maximum shared weight, in a similar fashion as in \cref{thm:ef:fpt_tw};
for polynomially-bounded weights where $W = (n+m)^{\bO(1)}$ this results in an algorithm of running time $(n+m)^{\bO(\tw)}$.
Again, our algorithm works even for the more general {\EFXmc} problem, and for multigraphs with non-identical valuations.

We first state the following observation.

\begin{observation}\label{obs:efx_zero_value_goods}
    Let an instance of {\EFXo} where vertices~$i,j$ share more than one edges of weight~$0$ for both endpoints.
    Then, removing all but one of those edges results in an equivalent instance.
\end{observation}

To see why \cref{obs:efx_zero_value_goods} holds, notice that any EFX orientation that does not orient all such edges to a \emph{single} endpoint,
can be modified to an EFX orientation that does.

\begin{theoremrep}[\appsymb]\label{thm:EFX:tw}
    Given an instance $\mathcal{I}$ of \EFXmc\ along with a nice tree decomposition of the input multigraph of width $\tw$,
    one can decide $\mathcal{I}$ in time $W^{\bO(\tw)} (n + m)^{\bO(1)}$,
    where $W$ denotes the maximum shared weight of $\mathcal{I}$ and $n$ the number of agents.
\end{theoremrep}



\begin{proof}
    Let $\mathcal{I} = (G,\{v_i\}_{i\in V})$ denote the instance of {\EFmc}
    and $(T,\{B_t\}_{t\in V(T)})$ the given nice tree decomposition of $G$ of width $\tw$.

    We will use \cref{obs:efx_zero_value_goods} as well as the following one for the correctness of our algorithm:

    \begin{observation}\label{obs:only_one_envious}
        Let $O$ be an orientation.
        If vertex $i$ envies, but not strongly envies, vertex $j$, then
        all edges oriented towards $j$ in $O$ are in $E_{ij}$. This additionally implies that no vertex $v\neq i$ envies $j$ in $O$.
    \end{observation}

    \textbf{Notations and Definitions.}
    Before describing our algorithm we first fix some notations. For $t\in V(T)$, we define $V_t$, $E_t$, $G_t$ and $A_t$ in the same way as in the proof of \cref{thm:ef:fpt_tw}.

    We define $A_t^{ef}$ to be the set of all partial orientations of $G_t$ where, first, all vertices of $V_t$ are  non strongly envious, and second, all vertices of $B_t$ which are envied (but not strongly) by some vertex of $V_t$, are also non strongly envious. Intuitively, those are the partial orientations of $G_t$ that can be completed into an envy-free orientation up to any good of $G$.

    \textbf{Signature definition.} Let $O\in A_t$. We define the \emph{signature} of $O$ with respect to $t$ as a triple $\sigma=(r,d,k)$ with $r\colon B_t\to [0,W]\cup\{None\}$, $d\colon B_t\to [0,W]\cup\{Some\}$ and $k\in [0,|E|]$ defined by:
    \begin{itemize}
        \item for all $i\in B_t$, $r(i)=\begin{cases}
            None & if\ X_i=\varnothing \\
            \min\{v_i(X_i),W\} & otherwise \\
        \end{cases}$

        \item for all $i\in B_t$, $d(i)=\begin{cases}
            Some & if\ there\ is\ \ell\in V_t\ envious\ of\ i\\
            \max_{\ell\in V_t,\ e\in X_\ell}\{v_i(X_\ell\setminus\{e\})\} & otherwise \\
        \end{cases}$

        \item $k=|E_t|-\Sigma_{i\in B_t\cup V_t}|X_i|$
    \end{itemize}

    Let $\Sigma_t$ denote the set of all possible signatures for node $t$, and notice that $|\Sigma_t|\leq |E|(W+2)^{2(tw+1)}$. We define the $\emph{signature function}$ $\sgn_t\colon A_t\to \Sigma_t$ that attributes to each partial orientation of $G_t$ its signature with respect to $t$.

    Take $O\in A_t^{ef}$. As in the algorithm of \cref{thm:ef:fpt_tw}, $sgn_t(O)=(r,d,k)$ has two types of information: $r$ and $d$ are the minimum amount of information needed to complete $O$ into a partial envy-free orientation up to any good: $r(i)$ represents what $i$ already \emph{received} from $O$ up to any good, and if $r(i)=None$, it means that $i$ did not receive any edge yet, and therefore from \cref{obs:EFX_one_good}, it means that in an extension of $O$, one vertex of $V\setminus V_t$ can be envious of $i$ without being strongly envious (note that we can have $r(i)=0$, as edges can have a value of $0$). $d(i)$ represents the amount that $i$ needs to receive to not be strongly envious of the vertices of $V_t$, and if $d(i)=Some$, it means that some vertex $\ell\in V_t$ is envious of $i$ (but not strongly envious), and therefore $i$ cannot receive any other edge, as otherwise $\ell$ would become strongly envious of $i$.

    In the following, we will allow ourselves to use the value $r(i)$ in additions and inequalities, even if it is possible that $r(i)=None$. In that case, we consider that $None$ has value $0$.

    \textbf{Signature valuation definition.} We define the \emph{valuation function} $val_t\colon \Sigma_t\to \{0,1\}$ defined by: for all $\sigma\in \Sigma_t$, $val_t(\sigma)=1$ if and only if there is $O\in A_t^{ef}$ such that $\sgn_t(O)=\sigma$.

    Take $t$ to be the root of $T$; by definition of a nice tree decomposition, $B_t$ is an empty set, $V_t=V$, $E_t=E$ and $G_t=G$. Therefore, $A_t^{ef}$ is exactly the set of all partial envy-free up to any good orientations of $G$. So, if we have an efficient way to compute the values of $val_t$, we can efficiently solve \EFXmc. We present an algorithm that computes the values of $val_t$ for all $t\in V(T)$ dynamically over the nice tree decomposition.

    \textbf{In the following, we describe the algorithm.}

    Take $t\in V(T)$:

    \textbf{Leaf node.} If $t$ is a leaf node, then $B_t=\varnothing$, $V_t=\varnothing$ and $E_t=\varnothing$. Therefore, $A_t=A_t^{ef}=\{\varnothing\}$. Take $epy$ the only function from $\varnothing$ to $[0,W]$, $val_t((epy,epy,0))=1$ and for all $k\in [1,m]$, $val_t((epy,epy,k))=0$.

    \textbf{Introduce node.} If $t$ is an introduce node with a child $t'$ such that $B_t=B_{t'}\cup \{i\}$ for some $i\notin B_{t'}$, the following lemma allows us to compute the values of $val_t$ over $\Sigma_t$ dynamically:

    \begin{lemma}
        For all $\sigma = (r,d,k)\in \Sigma_t$,

        \[val_t(\sigma)=\begin{cases}
            0 & if\ r(i)\neq None\ or\ d(i)\neq 0 \\
            val_{t'}((r|_{B_{t'}},d|_{B_{t'}},k)) & otherwise\\
        \end{cases}\]
    \end{lemma}

    Intuitively, we extend $val_{t'}$ to $B_t$ to define $val_t$.

    \begin{nestedproof}
        We note that $V_t=V_{t'}$. Furthermore, by definition of a tree decomposition, $i$ has no edge in common with $V_t$, and so no edge incident with $i$ is included in $E_t$: $E_t=E_{t'}$ and $i$ is an isolated vertex in $G_t$.

        Therefore, let $\sigma=(r,d,k)\in \Sigma_t$. If $d(i)\neq 0$ or $r(i)\neq None$, then there is no $O\in A_t$ where $sgn_t(O)=\sigma$.

        Otherwise, $d(i)=0$, $r(i)=None$ and as $E_t=E_t'$, let $O=\{X_j\}_{j\in B_t\cup V_t}\in A_t$, $X_i=\varnothing$. Let $O'=O\setminus\{X_i\}$, $O'\in A_{t'}$, and $\sgn_t(O)=\sigma$ if and only if $sgn_{t'}(O')=(r|_{B_{t'}},d|_{B_{t'}},k)$.

        Finally, $O\in A_t^{ef}$ if and only if it makes all the vertices of $V_t$ and all the vertices of $B_t$ envied by vertices of $V_t$ non strongly envious; but, as $i$ is an isolated vertex of $G_t$, and $V_t=V_{t'}$, it is true if and only if $O'$ makes all the vertices of $V_{t'}$ and all the vertices of $B_{t'}$ envied by vertices of $V_{t'}$ non strongly envious, if and only if $O'\in A_t^{ef}$. Therefore, the lemma holds.
    \end{nestedproof}

    Concerning the complexity, supposing the values of $val_{t'}$ have already been computed, the calculation of each value of $val_t$ can be done in constant time; therefore the calculation of all the values can be done in time $|\Sigma_t|\leq m(W+2)^{2(tw+1)}$.

    \textbf{Forget node.}

    Let $i,j\in V$. From \cref{obs:efx_zero_value_goods}, we can suppose that $|E_{ij}|\leq 2W +1$. We define the function $L_{ij}^i\colon [0,W]^2\times [0,W]^2\times[0,W]\cup\{\infty\}\times[0,2W+1]\to \{0,1\} $ in the following way:
    $L((a_1,a_2),(b_1,b_2),c,k)=1$ if and only if there is a partial orientation  $O=\{X_i,X_j\}$ of the subgraph $(\{i,j\},E_{ij})$ such that $v_i(X_i)=a_1$, $v_i(X_j)=a_2$, $v_j(X_j)=b_1$, $v_j(X_i)=b_2$, $|E_{ij}\setminus (X_i\cup X_j)|=k$, and finally $\min_{e\in X_j}\{v_i(e)\}=c$. We can compute the values of $L_{ij}^i$ via simple dynamic programming.

    If $t$ is a forget node with a child $t'$ such that $B_{t'}=B_t\cup \{i\}$ for some $i\notin B_t$, the following lemma allows us to compute the values of $val_t$ over $\Sigma_t$ dynamically:

    \begin{lemma}
        For all $\sigma=(r,d,k)\in \Sigma_t$: $val_t((r,d,k))=1$ if and only if the following holds:

        There is $\sigma'=(r',d',k')\in \Sigma_{t'}$, for all $j\in B_t$ there are $(a_j^1,a_j^2,b_j^1,b_j^2,c_j,k_j)\in [0,W]^4\times [0,W]\cup\{\infty\}\times [0,2W+1]$ verifying one of the following three points:

        \begin{enumerate}
            \item \begin{itemize}
                \item $d'(i)=Some$.
                \item for all $j\in B_t$, $a_j^1=b_j^2=0$.
                \item for all $j\in B_t$, if $d'(j)=Some$ then $k_j=|E_{ij}|$.
                \end{itemize}

            \item There is $j'\in B_t$, $c\in [0,W]$ such that:
                \begin{itemize}
                \item $r'(i)=None$.
                \item $d'(i)\neq Some$.
                \item $d'(j')\neq Some$.
                \item $L_{ij'}^{j'}((b_{j'}^1,b_{j'}^2),(a_{j'}^1,a_{j'}^2),c,k_{j'})=1$.
                \item for all $j\in B_t\setminus \{j'\}$, $a_j^1=b_j^2=0$.
                \item for all $j\in B_t\setminus \{j'\}$, if $d'(j)=Some$ then $k_j=|E_{ij}|$.
                \item  $d(j')=\max\{d'(j'),b_{j'}^2-c\}$.
                \end{itemize}
            \item \begin{itemize}
                \item $d'(i)\neq Some$.
            \end{itemize}
        \end{enumerate}

        and the following:
        \begin{itemize}
            \item $val_{t'}(\sigma')=1$.
            \item $k=k'+\sum_{j\in B_t}k_j$
            \item for all $j\in B_t$, $L_{ij}^i((a_j^1,a_j^2),(b_j^1,b_j^2),c_j,k_j)=1$
            \item for all $j\in B_t$, \begin{itemize}
                \item if $d(j)\neq Some$: $d'(j)\neq Some$, $d(j)=\max\{d'(j),b_j^2\}$ (except if this value has already been defined in case 2.), and $a_j^2\leq r'(i)+\sum_{j\in B_t}a_{j}^1$, and if $c_j\neq \infty$, $r(j)=\min\{r'(j)+b_j^1,W\}$, otherwise $r(j)=r'(j)$.
                \item if $d(j)=Some$ and $d'(j)\neq Some$: then $r(j)\leq b_j^1$, $r'(j)=None$, $r(j)=b_j^1$ and $a_j^2-c_j\leq r'(i)+\sum_{j\in B_t}a_{j}^1$.
                \item otherwise $d(j)=d'(j)=Some$, $c_j=\infty$, $r(j)=r'(j)$ and $r(j)\geq b_j^2$.
                \end{itemize}
            \item if $d'(i)\neq Some$, then $d'(i)\leq r'(i)+\sum_{j\in B_t}a_{j}^1$
        \end{itemize}

    \end{lemma}

    While the details of this lemma seem complex, the intuition behind it is fairly straightforward: when we have to treat vertex $i$, we have to choose the orientation of the edges it shares with $B_t$.

    In case \textbf{1.}, $i$ is envied by one vertex of $V_t$. Therefore, it cannot receive anything from any vertex of $B_t$, and all the edges between $i$ and $B_t$ must be oriented towards $B_t$.

    In case \textbf{2.}, we choose that $i$ is envied by one vertex $j$ of $B_t$. Therefore, it must not have received anything from any other vertex, we must orient the edges not shared with $j$ towards $B_t$.

    In case \textbf{3.}, $i$ is not envied by any vertex of $V_t$, and we choose that it won't be envied by any vertex of $B_t$ either. Therefore, we choose an orientation of the edges between $i$ and $B_t$ where $i$ will be non strongly envious of some vertices of $B_t$, and verify that it received enough to be non-strongly envious of the vertices of $V_t$.

    Finally, we check that the values of the two signatures and our chosen orientation correspond.

    \begin{nestedproof}
        We note that $V_t=V_t'\cup \{i\}$, and $E_t=E_{t'}\cup \bigcup_{j\in B_t} E_{ij}$. Also, $V_t\cup B_t=V_{t'}\cup B_{t'}$.

        In that case, a partial orientation $O\in A_t$ can be seen as a partial orientation $O'\in A_{t'}$, together with a partial orientation $O''$ of $(B_t',\bigcup_{j\in B_t} E_{i,j})$. In the following, we show that given a partial orientation $O\in A_t$, its signature $\sgn_t(O)$ can be computed using only the following information: the signature of $O'$ the restriction of $O$ to $G_{t'}$, and the orientation of $E_{ij}$ fo all $j\in V_t$ of $O$.

        Let $O=\{X_j\}_{j\in V_t\cup B_t}\in A_t$.

        For all $j\in B_t$, let $\alpha_j:=E_{ij}\cap X_i$, $\beta_j:=E_{ij}\cap X_j$, and $X'_j:=X_j\setminus E_{ij}$: $\alpha_j$ is the set edges of $E_{ij}$ received by $i$ while $\beta_j$ is the set of the ones received by $j$. $X'_j$ is the set of edges received by $j$ from $G_{t'}$.

        Let $X'_i:=X_i\setminus (\bigcup_{j\in B_t}E_{ij})$ and for all $\ell\in V_{t'}$, let $X'_\ell:=X_\ell$.  $O':= \{X'_j\}_{j\in B_{t'}\cup V_{t'}}$ is a partial orientation of $G_{t'}$, the restriction of $O$ to $G_{t'}$, so $O'\in A_{t'}$. Let $\sigma=(r,d,k)=\sgn_t(O)$ and $\sigma'=(r',d',k')=\sgn_{t'}(O')$.

        We note that $k=k'+\sum_{j\in B_t}|E_{ij}\setminus(\alpha_j\cup\beta_j)|$.

        For all $j\in B_t$, let $a^j_1:= v_i(\alpha_j)$, $a^j_2:= v_i(\beta_j)$, $b^j_1:= v_j(\alpha_j)$, $b^j_2:= v_j(\beta_j)$, and $c_j=\min_{e\in \beta_j}\{v_i(e)\}$. By linearity of the value functions, we have: $v_i(X_i)=v_i(X'_i)+\sum_{j\in B_t}a_j^1$, and for all $j\in B_t$, $v_j(X_j)=v_j(X'_j)+b^j_1$, $v_i(X_j)=v_i(X'_j)+a_2^j=a_2^j$, $v_j(X_i)=v_j(X'_i)+b_2^j=b_2^j$ and $k_j=|E_{ij}\setminus(\alpha_j\cup\beta_j)|$. By definition, $L_{ij}^i((a_j^1,a_j^2),(b_j^1,b_j^2),c_j,k_j)=1$ and $k=k'+\sum_{j\in B_t}k_j$.

        If $O\in A_t^{ef}$, then it follows immediately that $O'\in A_{t'}^{ef}$, and therefore, $val_{t'}(\sigma')=1$. Furthermore, $O$ makes $i$ non-strongly envious of any other vertex. Finally, from \cref{obs:only_one_envious}, either $i$ is envied (but not strongly) by a unique vertex of $V_t$ and therefore only receives edges shared with that vertex, or no vertex of $V_{t'}$ envies $i$, and either it receives edges from a unique vertex of $B_t$, or it doesn't. We will treat those three cases separately, which will result in the three cases of the lemma.

        First, suppose that $i$ is envied (but not strongly) by a vertex of $V_t'$ in $O$. Then, it also holds in $O'$. Therefore, $d'(i)=Some$. We show that it implies point 1.

        In that case, $i$ does not receive any edges shared with vertices of $B_t$, so, for all $j\in B_t$, $a_j^1=b_j^2=0$. Let $j\in B_t$, if $k_j<|E_{ij}|$, $j$ receives at least one edge from $i$ in $O$, and therefore, cannot be envied by a vertex of $V_{t'}$, so $d'(j)\neq Some$.

        Second, suppose that $i$ only receives edges shared with a vertex $j'\in B_t$. Then, it does not receive any edges shared with vertices of $V_{t'}$. Therefore, $r'(i)=None$ and $d'(i)\neq Some$. We show that it corresponds to case 2. Let $c=\min_{e\in \alpha_{j'}}\{v_{j'}(e)\}$. By definition, $L_{ij'}^{j'}((b_{j'}^1,b_{j'}^2)(a_{j'}^1,a_{j'}^2),c,k_{j'})=1$.

        For all $j\in B_t\setminus\{j'\}$, $i$ does not receive any edge shared with $j$, and therefore $a_{j}^1=b_{j}^2=0$. Also, if $k_j<|E_{ij}|$, $j$ receives at least one edge from $i$ in $O$, and therefore, cannot be envied by a vertex of $V_{t'}$, so $d'(j)\neq Some$. $i$ did not receive any edge from any vertex other than $j'$, and therefore $d(j')=\max\{d'(j'),b_{j'}^2-c\}$. Finally, $j'$ receives at least one edge from $i$, and therefore $d'(j)\neq Some$.

        Third, suppose that we are not in the two previous cases. Then, $d'(i)\neq Some$. It implies case 3 trivially.

        We prove the remaining points: let $j\in B_t$:

        If $d(j)\neq Some$: then $d'(j)\neq Some$ by definition, and if $i$ received edges from vertices other than $j$ in $O$, then $d(j)=\max\{d'(j),b_j^2\}$. Furthermore, $i$ is non-envious of $j$, so $a_j^2\leq r'(i)+\sum_{j\in B_t}a_{j}^1$; finally, if $c_j\neq\infty$, then $j$ receives at least one edge from $E_{ij}$ and therefore $r(j)=\min\{r'(j)+b_j^1,W\}$; otherwise, $r(j)=r'(j)$.

        If $d(j)=Some$ and $d'(j)\neq Some$: then $i$ is envious of $j$ but not strongly envious; therefore, $j$ does not receive any edge shared with any vertex of $V_{t'}$, so $r'(j)=None$ and $r(j)=b_j^1$. Furthermore, $j$ is non-envious of $i$, so $r(j)\leq b_j^1$. Finally, $i$ is non strongly envious of $j$, so $a_j^2-c_j\leq r'(i)+\sum_{j\in B_t}a_{j}^1$.

        Otherwise, if $d(j)=d'(j)=Some$, then $j$ is envied (but not strongly) by a vertex of $V_{t'}$. Therefore, $j$ does not receive any edge of $E_{ij}$, so $c_j=\infty$ and $r(j)=r'(j)$. Also, $j$ does not envy $i$, and therefore $r(j)\geq b_j^2$.

        Finally, $i$ is not strongly envious of the vertices of $V_{t'}$, and therefore, if $d'(i)\neq Some$, $d'(i)\leq r'(i)+\sum_{j\in B_t}a_{j}^1$.

    \end{nestedproof}

    \textbf{Join node.} If $t$ is a join node with two children $t_1,t_2$ such that $B_t=B_{t_1}=B_{t_2}$, the following lemma allows us to compute the values of $val_t$ over $\Sigma_t$ dynamically:

    \begin{lemma}
        For all $\sigma=(r,d,k)\in \Sigma_t$: $val_t((r,d,k))=1$ if and only if there are $\sigma_1=(r_1,d_1,k_1)\in \Sigma_{t_1}$ and $\sigma_2=(r_2,d_2,k_2)\in \Sigma_{t_2}$ such that:

        \begin{itemize}
            \item $val_{t_1}(\sigma_1)=1$ and $val_{t_2}(\sigma_2)=1$
            \item $k=k_1+k_2$
            \item for all $j\in B_t$,
            \begin{itemize}
                \item if $d_1(j)=Some$ (respectively $d_2(j)=Some$) then $r_2(j)=None$ (respectively $r_1(j)=None$), $d(j)=Some$, $d_2(j)\leq r_1(j)$ (respectively $d_1(j)\leq r_2(j)$) and $r(j)=r_1(j)$ (respectively $r(j)=r_2(j)$).
                \item otherwise, if $r_1(j)=r_2(j)=None$ then $r(j)=None$ and $d(j)=\max\{d_1(j),d_2(j)\}$
                \item otherwise, $r(j)=\min\{r_1(j)+r_2(j),W\}$ and $d(j)=\max\{d_1(j),d_2(j)\}$
            \end{itemize}
        \end{itemize}

    \end{lemma}
    Intuitively, a valid partial orientation of $G_t$ is a combination of two valid partial orientations of $G_{t_1}$ and $G_{t_2}$; therefore, we can combine two possible valid partial orientations of $G_{t_1}$ and $G_{t_2}$ and see if they are compatible, that is, if some vertex of $B_t$ is envied by one vertex of $V_{t_1}$, it must not receive any edge from $G_{t_2}$, and conversely. In that case, we can combine the two orientations into an orientation of $G_t$.

    \begin{nestedproof}
        Note that $V_t=V_{t_1}\cup V_{t_2}$, with $V_{t_1}\cap V_{t_2}=\varnothing$. Similarly, $E_t=E_{t_1}\cup E_{t_2}$, with $E_{t_1}\cap E_{t_2}=\varnothing$.

    Let $O=\{X_j\}_{j\in\{B_t\cup V_t\}}\in A_t$; for all $j\in B_{t_1}\cup V_{t_1}$, let $X_j^1:=X_j\cap E_{t_1}$ and
    for all $j\in B_{t_2}\cup V_{t_2}$, let $X_j^2:=X_j\cap E_{t_2}$.

    Let $O_1:=\{X_j\}_{j\in\{B_{t_1}\cup V_{t_1}\}} $ and $O_2:=\{X_j\}_{j\in\{B_{t_2}\cup V_{t_2}\}}$. They are the restrictions of $O$ to $G_{t_1}$ and $G_{t_2}$ respectively: $O_1\in A_{t_1}$ and $O_2\in A_{t_2}$.

    Let $\sigma=(r,d,k)=\sgn_t(O)$, $\sigma_1=(r_1,d_1,k_1)=\sgn_{t_1}(O_1)$ and $\sigma_2=(r_2,d_2,k_2)=\sgn_{t_2}(O_2)$. By definition, $k=k_1+k_2$, $val_1(\sigma_1)=1$ and $val_2(\sigma_2)=1$.

    Let $j\in B_t$. If $d_1(j)=Some$, some vertex of $V_{t_1}$ envies $j$, so $d(j)=Some$. Furthermore, in that case, $j$ does not receive any edge shared with any other vertex, and so $r_2(j)=None$, and $r(j)=r_1(j)$. Finally, $j$ does not envy any vertex from $V_{t_2}$, so $d_2(j)\leq r_1(j)$. The same holds symmetrically if $d_2(j)=Some$.

    Otherwise, if $r_1(j)=r_2(j)=None$, then $j$ does not receive any edge shared with vertices of $V_t$, and therefore $r(j)=None$. Furthermore, no vertex of $V_t$ is envious of $j$, so $d(j)=\max\{d_1(j),d_2(j)\}$.

    Finally, otherwise, $j$ receives at least one edge shared with vertices of $V_t$, so $r(j)=\min\{r_1(j)+r_2(j),W\}$, and no vertex of $V_t$ is envious of $j$, so $d(j)=\max\{d_1(j),d_2(j)\}$.
    \end{nestedproof}

    Concerning the complexity, supposing the values of $val_{t_1}$ and $val_{t_2}$ have already been computed, there are $m\times W^{\bO(\tw)}$ signatures of $t_1$ and $t_2$ to consider. Therefore, this step can be done in time $mW^{\bO(tw)}\times mW^{\bO(tw)}=m^2W^{\bO(tw)}$.
    \medskip

    This concludes the description and correctness of our algorithm. We now argue about its complexity: as the nice tree decomposition can be supposed to be of polynomial size and each step can be done in time at most $m^2W^{\bO(\tw)}$, the algorithm has complexity $W^{\bO(\tw)}(nm)^{\bO(1)}$.
\end{proof}

\begin{proofsketch}
    The algorithm closely follows the one of \cref{thm:ef:fpt_tw}.
    For a node $t$ of the given nice tree decomposition, let $G_t$ denote the corresponding subgraph,
    while $B_t$ denotes the set of vertices appearing in the bag of node $t$.
    Once again, for every node $t$ of the tree decomposition, we keep track of whether there exists a partial orientation of $G_t$
    (excluding the edges among vertices of $B_t$)
    that guarantees that the vertices of $V(G_t) \setminus B_t$ do not strongly envy the vertices of $V(G_t)$,
    while also keeping track of the number of unoriented edges of the partial orientation
    as well as the following information for each of the at most $\tw+1$ vertices in $B_t$:
    \begin{itemize}
        \item its \emph{revenue} due to edges shared with vertices in $V(G_t) \setminus B_t$,
        separately handling the case where a vertex has no received goods,

        \item its \emph{demand} so that it does not strongly envy any vertex of $V(G_t) \setminus B_t$,
        separately handling the case where a vertex is non-strongly envied by a vertex of $V(G_t) \setminus B_t$.
    \end{itemize}
    Notice that both the revenue and the demand of a vertex in $B_t$ can be upper-bounded by $W$,
    as a vertex with revenue $W$ cannot envy any other vertex, while a vertex may envy another up to a value of at most $W$.
    Lastly, by \Cref{obs:efx_zero_value_goods} we can assume that for every pair of vertices there is at most one edge of valuation $0$
    for both endpoints,
    thus there are at most $2W + 1$ edges shared between them, implying that
    the number of unoriented edges is at most $(2W+1) \cdot \binom{n}{2}$.
    Consequently, the total running time is $W^{\bO(\tw)} (n+m)^{\bO(1)}$.
\end{proofsketch}

Just like in \cref{subsec:EFo_shared_weight}, we obtain the following.

\begin{corollary}\label{cor:EFXmc_poly_weights_algo}
    There is an algorithm running in time $(n+m)^{\bO(\tw)}$ for \EFXmc\ with polynomially-bounded weights.
\end{corollary}

\begin{remark}
    The algorithm of \cref{cor:EFXmc_poly_weights_algo} matches the lower bound of $n^{o(\vc)}$ in \cref{cor:EFXo_W1}.
\end{remark}

\section{Conclusion}\label{sec:conclusion}

In this work we initiated the study of EF orientations in multigraphs and presented a variety of both positive and negative results regarding
their computational complexity, including a complete characterization for the case of binary valuations. We additionally studied them under the perspective of parameterized complexity,
and showed how some of the results can even be transferred to the related {\EFXo} problem, improving upon previous work on the latter in the process and settling an open question by Deligkas et al.~\cite{ijcai/DeligkasEGK25}. Our work demonstrates that EF orientations are both interesting in their own right, and useful for transferring results to EFX orientations.

As a direction for future work, we comment that \cref{thm:EFo:Heavy} generalizes \cref{thm:EFo:linear_algo}, albeit only for simple graphs.
A natural question is thus whether we can develop such an algorithm that works for multigraphs as well, that is, whether {\EFo}
admits an FPT algorithm parameterized by the number $k$ of heavy edges running in time $2^k n^{\bO(1)}$.
Another interesting direction concerns our parameterized algorithms in \cref{subsec:EFo_shared_weight,subsec:EFXo:tw}, specifically, whether our parameterization by treewidth plus maximum shared weight can be improved to one by treewidth plus \emph{maximum weight}.

One last future research direction we would like to highlight is the search for regimes in which the complexity of EF and EFX \textsc{Orientation} differs.
We showed that one such case regards instances with binary valuations, while the two problems have similar behavior in the case of
polynomially-bounded weights and small vertex cover/treewidth.
In which other (valuation or structural) regime does the complexity of the two problems differ?

\subsubsection*{Acknowledgements}

We are grateful to our PhD advisors, Michael Lampis, Aris Pagourtzis, and Alkmini Sgouritsa, for valuable discussions and guidance regarding the problems studied in this work.

S.K., C.P., and M.S. were partially supported by project MIS 5154714 of the National Recovery and Resilience Plan Greece 2.0 funded by the European Union under the NextGenerationEU Program.

E.N. and M.V. were partially supported by the ANR project ANR-21-CE48-0022 (S-EX-AP-PE-AL).

S.K. and M.V. were partially supported by the project COAL-GAS,
financed under the Global Seed Fund program of Université Paris Sciences et Lettres (PSL) and by the CNRS.

S.K., C.P., and M.V. were partially supported by the Partenariat Hubert Curien France – Grèce 2025 programme ``Aristote'', funded by the Ministry of Europe and Foreign Affairs (MEAE) and the Ministry of Higher Education, Research, and Innovation (MESRI) of France, and the State Scholarships Foundation (IKY) of Greece, via the Cooperation and Cultural Action Service of the French Embassy, and the French Institute of Greece.

M.V. was partially supported by the ERC grant PARAMLP (project no. 101220563) funded by the European Union.
Views and opinions expressed are however those of the author(s) only and do not necessarily reflect those of the European Union or the European Research Council Executive Agency. Neither the European Union nor the granting authority can be held responsible for them.

S.K., C.P., and M.S. acknowledge the financial support of the Archimedes Research Unit (Athena RC) during their PhD.

No artificial intelligence was used at any point during the production of this paper. All of the presented results were conceived exclusively by the authors and the entirety of the paper was written exclusively by the authors.






\bibliographystyle{plainurl}
\bibliography{bibliography}

\end{document}